\documentclass[draft,onecolumn,11pt,titlepage]{IEEEtran}
\usepackage[cmex10]{amsmath}
\usepackage{amsfonts}
\usepackage[final]{graphics}
\usepackage[final]{graphicx}
\usepackage{amsbsy}
\usepackage{amsbsy}
\usepackage{amssymb}
\usepackage{url}
\usepackage{psfrag}
\usepackage{color}
\usepackage{enumerate}

\newtheorem{theorem}{Theorem}

\numberwithin{equation}{section}

\newcommand{\eqdef}{\triangleq}
\newcommand{\bm}[1]{{\boldsymbol{#1}}}
\def\x{{\mathbf x}}
\def\y{{\mathbf y}}
\def\h{{\mathbf h}}

\def\bA{{\mathbf A}}
\def\bR{{\mathbf R}}

\newcommand{\s}{\mathbf{s}}
\newcommand{\C}{{\mathbb{C}}}
\newcommand{\R}{{\mathbb{R}}}
\newcommand{\N}{{\mathbb{N}}}
\newcommand{\E}{\mathrm{E}}

\newcommand{\ratee}{\beta_t^{i,2}}
\newcommand{\rate}{\beta_t^{i,1}}

\newcommand{\qed}{\nobreak \ifvmode \relax \else
\ifdim\lastskip<1.5em \hskip-\lastskip
\hskip1.5em plus0em minus0.5em \fi \nobreak
\vrule height0.75em width0.5em depth0.25em\fi}

\def\bdm#1\edm{\begin{displaymath}#1\end{displaymath}}
\def\be#1\ee{\begin{equation}#1\end{equation}}
\def\barr#1\earr{\begin{align}#1\end{align}}

\newcommand{\sfootnote}[1]{{\footnote{\setlength{\baselineskip}{4.0mm}#1}}}

\begin{document}

\title{\setlength{\baselineskip}{10.0mm}
Transmitting important bits and sailing high radio waves: a decentralized
cross-layer approach to \\ cooperative video transmission}

\author{
Nicholas Mastronarde, Francesco Verde, Donatella Darsena, Anna Scaglione, and Mihaela van der Schaar
\thanks{
\setlength{\baselineskip}{4.0mm}
N.~Mastronarde is with the Department of Electrical Engineering,
State University of New York at Buffalo, Buffalo, NY 14260, USA (e-mail: nmastron@buffalo.edu).
This work was done while he was at the
University of California at Los Angeles (UCLA), Los Angeles, CA 90095-1594, USA.

M.~van der Schaar is with the Department of Electrical Engineering, University of
California at Los Angeles (UCLA), Los Angeles, CA 90095-1594,
USA (e-mail: mihaela@ee.ucla.edu).

F.~Verde is with the Department
of Biomedical, Electronic and Telecommunication Engineering,
University Federico II, Naples I-80125, Italy (e-mail: f.verde@unina.it).

D.~Darsena is with the Department for Technologies, Parthenope University,
Naples I-80143, Italy (e-mail: darsena@uniparthenope.it)

A.~Scaglione is with the Department of Electrical
and Computer Engineering, University of California, Davis, CA
95616-5294, USA (e-mail: ascaglione@ucdavis.edu).

The work of N. Mastronarde and M. van der Schaar was supported in part by NSF grant no.
0831549.
}
}

\maketitle

\vspace{-20mm}

\begin{abstract}

\setlength{\baselineskip}{4.5mm}

We investigate the impact of cooperative relaying on
uplink and downlink multi-user (MU) wireless video transmissions.
The objective is to maximize the long-term sum of utilities across
the video terminals in a decentralized fashion, by jointly optimizing
the packet scheduling, the resource allocation, and the cooperation decisions,
under the assumption that some nodes are willing to act as cooperative relays.
A pricing-based distributed resource allocation
framework is adopted, where the price reflects the expected future congestion in the network.
Specifically, we formulate the wireless video transmission problem as
an MU Markov decision process (MDP) that explicitly considers the cooperation at the physical layer
and the medium access control sublayer,
the video users' heterogeneous traffic characteristics, the dynamically varying network conditions,
and the coupling among the users' transmission strategies across time due to the shared wireless resource.
Although MDPs notoriously suffer from the curse of dimensionality, our study shows that, with appropriate simplications and approximations, the complexity of the MU-MDP can be significantly mitigated.
Our simulation results demonstrate that integrating cooperative decisions into the MU-MDP optimization can increase the resource price in networks that only support low transmission rates and can decrease the price in networks that support high transmission rates. Additionally, our results show that cooperation allows users with feeble direct signals to achieve improvements in video quality on
the order of $5-10$ dB peak signal-to-noise ratio (PSNR), with less than 0.8 dB quality loss by users with strong direct signals, and with a moderate increase in total network energy consumption that is significantly less than the energy that a distant node would require to achieve an equivalent PSNR without exploiting cooperative diversity.

\end{abstract}

\vspace{-5mm}

\begin{keywords}
\setlength{\baselineskip}{4.5mm}
Cooperative communications, cross-layer optimization, decode-and-forward relaying,
Markov decision process (MDP), multi-user scheduling,
resource allocation, wireless video transmission.
\end{keywords}

\vspace{-5mm}
\section{Introduction}
\label{sec:intro}

Existing wireless networks provide dynamically varying resources with only limited support for the Quality of Service (QoS) required by delay-sensitive, bandwidth-intense, and loss-tolerant multimedia applications. This problem is further exacerbated in multi-user (MU) settings because they require multiple video streams, with heterogeneous traffic characteristics, to share the scarce wireless resources. To address these challenges, a lot of research has focused on MU wireless communication~\cite{Chiang07, Knopp95, Viswanath02, Tse05, Yu07} and, in particular, MU video streaming over wireless networks~\cite{Zhang07, Huang08, Maani08, Su07, Fu09a}.
%
The majority of this research relies on cross-layer adaptation to match available system resources (e.g., bandwidth, power, or transmission time) to application requirements (e.g., delay or source rate), and vice versa. In MU video streaming applications~\cite{Zhang07, Huang08, Maani08, Su07, Fu09a}, for example, cross-layer optimization is deployed to strike a balance between scheduling lucky users who experience very good fades, and serving users who have the highest priority video data to transmit. This tradeoff is important because rewarding a few lucky participants, as opportunistic multiple access policies do~\cite{Knopp95, Viswanath02, Tse05}, does not translate to providing good quality to the application (APP) layer.
%
Unfortunately, with the exception of~\cite{Yu07,Alay09}, the aforementioned research assumes that wireless users are
{\em noncooperative}. This leads to a basic inefficiency in the way that the network resources are assigned: indeed, good fades experienced by some nodes can go to waste because users with higher priority video data, but worse fades, get access to the shared wireless channel.

A way to not let good fades go to waste is to enlist the nodes that experience good fades as cooperative helpers, using a number of techniques available for cooperative coding~\cite{Laneman-1,Send-1,Laneman-2}. As mentioned above, this idea has been considered in~\cite{Yu07,Alay09}.
In~\cite{Alay09}, for example, a cross-layer optimization is proposed involving the
physical (PHY) layer, the medium access control (MAC) sublayer, and the APP layer,
where layered video coding is integrated with randomized cooperation to enable efficient video multicast in a cooperative wireless network. However, because it is a multicast system, there is no need for an optimal multiple-access strategy, and no need to worry about heterogeneous traffic characteristics.
In~\cite{Yu07}, a centralized network utility maximization (NUM) framework is proposed for jointly optimizing relay strategies and resource allocations in a cooperative orthogonal frequency-division multiple-access (OFDMA) network.
In both~\cite{Yu07,Alay09}, it is assumed that each user has a static utility function of the average transmission rate, where the utility derived by each user in~\cite{Alay09} is a function of the average received rate of the base and enhancement layer video bitstreams. 

Unlike the aforementioned solutions, we take a dynamic optimization approach to the cooperative MU video streaming problem.
In particular, unlike~\cite{Yu07,Alay09}, the solution that we adopt explicitly considers packet-level video traffic characteristics (instead of flow-level) and dynamic network conditions (instead of average case conditions).
Our solution is inspired by the cross-layer resource allocation and scheduling solution in~\cite{Fu09a}, in which the MU wireless video streaming problem is modeled and solved as an MU Markov decision process (MDP) that allows the users, via a uniform resource pricing solution, to obtain long-term optimal video quality in a distributed fashion.  
However, although we use the traffic model and dual decomposition proposed in \cite{Fu09a}, cooperation 
renders our PHY/MAC model  completely different from that studied in \cite{Fu09a}, thus 
opening additional research issues with respect to \cite{Fu09a}, such as how the cooperation decision should be made, what is the impact of cooperation on the resource price, and what is the impact of cooperation on the total network energy consumption.
Moreover,
as recently shown in \cite{Icassp}, augmenting the framework developed in \cite{Fu09a} to also account for cooperation is challenging because of the complexity of the resulting cross-layer MU-MDP optimization.

The contributions of this paper are fourfold.
%
First, we formulate the cooperative wireless video transmission problem as
an MU-MDP using a time-division multiple-access (TDMA)-like network, randomized space-time block coding (STBC) \cite{Sirkeci07}, and a decode-and-forward cooperation strategy. To the best of our knowledge, we are the first to consider cooperation in a dynamic optimization framework.
We show analytically that the decision to cooperate can be made opportunistically, independently of the MU-MDP. Consequently, each user can determine its optimal scheduling policy by only keeping track of its experienced cooperative transmission rates, rather than tracking the channel statistics throughout the network.
%
Second, in light of the fact that opportunistic cooperation is optimal, we propose a low complexity opportunistic cooperative strategy for exploiting good fades in an MU wireless network. The key idea is that nodes can, in a distributed manner, self-select themselves to act as cooperative relays. The proposed self-selection strategy requires a number of message exchanges that is linear in the number of video sources, and selects sets of cooperative relays in such a way that cooperation can be guaranteed to be better than direct transmission.
%
Third, we show experimentally that users with feeble direct signals to the access point (AP) are conservative in their resource usage when cooperation is disabled. In contrast, when cooperation is enabled, users with feeble direct signals to the AP use cooperative relays and utilize resources more aggressively. Consequently, the uniform resource price that is designed to manage resources in the network tends to increase when cooperation is enabled in a network that only supports low transmission rates, but tends to decrease when it is enabled in a network that supports high transmission rates.
%
Fourth, we study the impact of cooperation on the total network energy consumption. We show that the increased transmission rate afforded by cooperation requires an increase in total network energy relative to the lower rate direct transmission; however, this increase is moderate compared to the amount of power required to transmit directly to the access point at a transmission rate equivalent to the cooperative rate.

The remainder of the paper is organized as follows. We introduce the system and application models in Section~\ref{sec:system-model}. In Section~\ref{sec:coopPHY}, we provide expressions for the transmission rate, packet error rate, and network energy consumption in both direct and cooperative transmission modes.
In Section~\ref{sec:opt-scheduling}, we present the proposed MU cross-layer PHY/MAC/APP optimization. In Section~\ref{protocol}, we propose a distributed protocol for opportunistically recruiting cooperative relays. Finally, we report numerical results in Section~\ref{sec:results} and conclude in Section~\ref{sec:concl}.

\vspace{-5mm}
\section{System Model}
\label{sec:system-model}

We consider a network composed of $M$ users streaming video content over a shared
wireless channel to a single AP (see Fig.~\ref{fig:network}).
Such a scenario is typical of many uplink media applications, such as remote monitoring and surveillance, wireless
video sensors, and mobile video cameras. The proposed
optimization framework can also be used for downlink applications, where the
relays can be recruited for streaming video to a certain user in the network in exactly the same way that they can be recruited to transmit to the AP in the uplink scenario.
In Subsection~\ref{sec:mac-phy-model}, we introduce the MAC and PHY layer models. Then, in Subsection~\ref{sec:app-model}, we describe the deployed APP layer model.

\vspace{-5mm}
\subsection{MAC and PHY layer models}
\label{sec:mac-phy-model}

We assume that time is slotted into discrete time-intervals of length $R>0$ seconds and
each time slot is indexed by $t \in \N$.\sfootnote{
The fields of complex, real, and nonnegative integer numbers are
denoted with $\mathbb{C}$, $\mathbb{R}$, and $\mathbb{N}$, respectively;
matrices [vectors] are denoted with upper [lower] case boldface
letters (e.g., $\bA$ or $\x$);
the field of $m \times n$
complex [real] matrices is denoted
as $\C^{m \times n}$ [$\R^{m \times n}$],
with $\C^m$ [$\R^m$] used as a
shorthand for $\C^{m \times 1}$ [$\R^{m \times 1}$];
the superscript $T$ denotes the transpose of a vector;
$\left| \cdot \right|$ denotes the magnitude of a complex number;
$\|\x\|_1$ is the $l_1$ norm of the vector $\x \in \C^n$, which
for positive real-valued vectors is simply the sum of the
components, whereas $\|\x\|_2$ is the Euclidean norm of $\x \in \C^n$;
$\{\bA \}_{ij}$ indicates the $(i+1,j+1)$th element of
the matrix $\bA \in \C^{m \times n}$, with
$i \in \{0,1, \ldots, m-1\}$
and $j\in\{0,1,\ldots,n-1\}$;
a circular symmetric complex Gaussian random variable $X$ with mean $\mu$ and variance $\sigma^2$ is denoted as $X \sim {\cal CN}(\mu, \sigma^2)$;
$\lfloor \cdot \rfloor$ and $\lceil \cdot \rceil$ denote flooring- and ceiling-integer, respectively;
$\E[\cdot]$ stands for ensemble averaging; and, finally, $[\cdot]^+ = \max(\cdot,0)$.
}
At the MAC sublayer, the users access the shared channel using a TDMA-like protocol.
In each time slot $t$, the AP endows the $i$th user, for $i \in \{1,2,\cdots,M\}$, with the resource fraction $x^i_t$,
where $0 \le x^i_t \le 1$, such that the user can use the amount of channel time $R \, x^i_t$ for transmission.
Let $\x_t \eqdef (x^1_t, x^2_t, \ldots, x^M_t)^T \in \R^M$ denote the resource allocation vector at time slot $t$,
which must satisfy the stage resource constraint $\|\x_t\|_1 =\sum_{i=1}^M x_t^i \leq 1$, where the inequality accounts for possible signaling overhead. 

Each node's PHY layer is assumed to be a single-carrier single-input single-output system designed to handle
quadrature amplitude modulation (QAM) square constellations,
with a (fixed) symbol rate of $1/T_s$ symbols per second.
The PHY layer can support a set of $N+1$
data rates $\beta_{n} \eqdef b_{n}/T_{\text{s}}$ (bits/second), where $b_{n} \eqdef \log_2(M_{n})$ is the number of
bits that are sent every symbol period, with $n \in \{0, 1, \ldots,N\}$, and $M_{n}$ is the number of signals in the QAM constellation.
Hence, $\beta_{0} \le \beta_{1} \le \cdots \le \beta_{N}$  form the {\em basic rate set} $\mathcal{B}$ and
$\beta_{0}$ is the {\em base rate} at which the nodes exchange control messages.
Let $d_n$ be the minimum distance of the $M_n$-QAM constellation,
the average transmitter energy per symbol is given by
\be
\mathcal{E}_s \eqdef d_n^2 \left(\frac{M_{n}-1}{6}\right) \: \text{ (Joules) },
\ee
which is assumed to be fixed for all the nodes and data rates, i.e., 
it does not depend on the indices $i$ and $n$.
Consequently, the average power per symbol expended by each transmitter 
is $\mathcal{P}_s \eqdef \mathcal{E}_s/T_s$ (Watts).
We consider a frequency non-selective
block fading model, where $h_t^{i\ell} \in \C$ denotes the fading coefficient
over the $i \rightarrow \ell$ link in time slot $t$,
with $i \neq \ell \in \{0, 1, \ldots, M\}$, and
$i=0$ or $\ell=0$ corresponding to the AP.
It is assumed that all the channels are dual, i.e., $|h_t^{i\ell}|=|h_t^{\ell i}|$,
and that the fading coefficients $h_t^{i\ell}$ are independent and 
identically distributed (i.i.d.) with respect to $t$. 
Moreover, we define ${\bf H}_t \in \C^{M \times M}$
as the matrix collecting the fading coefficients among all of the nodes and the AP,
i.e., $\{{\bf H}_t\}_{i\ell} = h_t^{i\ell}$, for $i \neq \ell \in \{0, 1,\ldots, M\}$.

At the PHY layer, there are two transmission modes to choose from: direct and cooperative. In the {\em direct} transmission mode, as shown in Fig.~\ref{fig:network}, the $i$th source node transmits directly to the AP at the data rate $\beta_t^{i 0} \in \mathcal{B}$ (bits/second) for the assigned transmission time of
$R \, x^i_t$ seconds. In the {\em cooperative} transmission mode, some nodes serve as decode-and-forward relays. Specifically, in the cooperative mode, the assigned transmission time is divided into two phases as illustrated in Fig.~\ref{fig:network}: in {\em Phase I}, the $i$th source node directly broadcasts its own data to all the nodes in the network at the data rate $\rate \in \mathcal{B}$ for $R \, \rho_t^i \, x_t^i$ seconds, where $0 < \rho_t^i <1$ is the Phase I time fraction; 
in {\em Phase II}, some of the nodes overhearing the source transmission, belonging to a certain subset  ${\cal C}_t^i \subseteq \{1, 2, \ldots, M\}-\{i\}$, demodulate the data received in Phase I, re-modulate
the original source bits, and then cooperatively transmit towards the AP, along with the original source
$i$, at the data rate $\ratee \in \mathcal{B}$ for the remaining $R \, (1-\rho_t^i) \, x_t^i$ seconds.
In the sequel, we denote with $\beta_t^{i,\text{coop}}$ (bits/second)
the {\em cooperative data rate} over the two phases, i.e., 
the amount of bits that are transmitted in a single phase 
divided by the overall length of the two phases, which depends on 
the data rates $\rate$ and $\ratee$ attainable in each of the two hops.
The decision to transmit in the direct or cooperative transmission mode depends
on fading coefficients throughout the network in time slot $t$ and on the target
packet error rate (PER).
Thus, the actual transmission rate of the $i$th source in time slot $t$ is dictated by
the cooperation decision $z_t^i \in \{0,1\}$, where $z_t^i = 1$ if cooperation
is chosen, and $z_t^i = 0$ if direct transmission is chosen.
In Section~\ref{sec:coopPHY}, we compute the transmission parameters $\beta_t^{i 0}$ and 
$\beta_t^{i,\text{coop}}$ as functions of a subset of the entries in ${\bf H}_t$,  as well as the
time fraction $\rho_t^i$, and, in Section~\ref{protocol}, we describe how to determine the 
set of cooperative relays ${\cal C}_t^i$ and the cooperation decision $z_t^i$.

\vspace{-5mm}
\subsection{APP layer model and packet scheduling}
\label{sec:app-model}


The source traffic can be modeled using any Markovian traffic model (e.g. \cite{Fu09a, Salodkar10}). However, to accurately capture the characteristics of the video packets, we adopt the sophisticated video traffic model proposed in \cite{Fu09a}, which accounts for the fact that video packets have different deadlines, distortion impacts, and source-coding dependencies (whereas the model in \cite{Salodkar10} does not consider these characteristics). In this section, we describe the key features of this model, but because the problem formulation and novelty of this paper do not depend on the deployed traffic model (so long as the model is Markovian), we refer the interested reader to \cite{Fu09a} for complete details.

For $i \in \{1, 2, \ldots, M\}$, the {\em traffic state} ${\cal T}_t^i \eqdef \{{\cal F}_t^i, {\bf b}_t^i\}$ represents the video data that the $i$th user can potentially transmit in time slot $t$, and comprises the following two components: the schedulable frame set ${\cal F}_t^i$ and the buffer state ${\bf b}_t^i$. In time slot $t$, we assume that the $i$th user can transmit packets belonging to the set of video frames ${\cal F}_t^i$ whose deadlines are within the scheduling time window (STW) $[t,t+W]$. The buffer state ${\bf b}_t^i \eqdef (b_{t,j}^i \,|\, j \in {\cal F}_t^i )^T$ represents the number of packets of each frame in the STW that are awaiting transmission at time $t$. The $j$th component $b_{t,j}^i$ of ${\bf b}_t^i$ denotes the number of packets of frame $j \in {\cal F}_t^i$ remaining for transmission at time $t$.  We assume that each packet has size $P$ bits.
Fig.~\ref{fig:traffic_state} illustrates how the traffic states are defined for a simple IBPB GOP structure.\sfootnote{In a typical hybrid video coder like H.264/AVC or MPEG-2, I, P, and B indicate the type of motion prediction used to exploit temporal correlations between video frames. I-frames are compressed independently of the other frames, P-frames are predicted from previous frames, and B-frames are predicted from previous and future frames.}

We now define the packet scheduling action. In each time slot $t$, the $i$th user takes scheduling action $\y_t^i \eqdef (y_{t,j}^i \, | \, j \in {\cal F}_t^i)^T$, which determines the number of packets to transmit out of ${\bf b}_t^i$.
Specifically, the $j$th component $y_{t,j}^i$ of $\y_t^i$ represents the number of packets of the $j$th frame within the STW that are scheduled to be transmitted in time slot $t$.
Importantly, the scheduling action ${\bf y}^i_t$ is constrained to be in the feasible scheduling action set
${\cal P}^i({\cal T}_t^i,\beta_t^i)$, which depends on the traffic state ${\cal T}_t^i$ and the transmission rate supported by the PHY layer $\beta_t^i$.
In particular, the following three constraints must be met:

\begin{enumerate}

\itemsep=0mm

\item
{\em Buffer}: Every component of ${\bf y}^i_t$ must satisfy $0 \le y_{t,j}^i \le b_{t,j}^i$.

\item
{\em Packet}: The total number of transmitted packets must satisfy $\|{\bf y}^i_t\|_1 =\sum_{j \in {\cal F}_t^i} y_{t,j}^i \leq
\frac{R \, \beta_t^i}{P}$,
where $\beta_t^i=\beta_t^{i 0}$ in the direct transmission mode, i.e., when $z_t^i=0$, and
$\beta_t^i=\beta_t^{i,\text{coop}}$ in the cooperative transmission mode, i.e., when $z_t^i=1$.
Note that $\beta_t^i$ depends on a subset of the elements in ${\bf H}_t$ as described later in Section~\ref{sec:coopPHY}.\sfootnote{We do not include $x_t^i$ in the packet constraint $\|{\bf y}^i_t\|_1 =\sum_{j \in {\cal F}_t^i} y_{t,j}^i \leq \frac{R \, \beta_t^i}{P}$ because $x_t^i$ is not known at the time the scheduling decision ${\bf y}^i_t$ is determined. Once the scheduling decision is determined, the resource allocation $x_t^i$ is determined as $x_t^i = \frac{P}{R  \beta_t^i}   \|{\bf y}_t^i\|_1$ (see ~\eqref{eq:sched2rsrc}). Importantly, the stage resource constraint ensures that the scheduling decisions ${\bf y}_t^i, \forall i \in \{1,\ldots,M\}$, are selected such that $\sum_{i=1}^M x_t^i \leq 1$.}

\item
{\em Dependency}: If there exists a frame $k$ that has not been transmitted, and frame $j$ depends on frame $k$ (denoted by $k \prec j$), then $\left(b_{t,k}^i - y_{t,k}^i\right)y_{t,j}^i=0$. In other words, all packets associated with $k$ must be transmitted before transmitting any packets associated with $j$.

\end{enumerate}

The sequence of traffic states $\{ {\cal T}_t^i : t \in \N \}$ can be modeled as a controllable Markov chain with transition probability function
$p({\cal T}_{t+1}^i \, | \, {\cal T}_t^i, {\bf y}_t^i)$.

\vspace{-5mm}
\section{Cooperative PHY layer transmission}
\label{sec:coopPHY}

In this subsection, with reference to the uplink scenario,
we describe how the direct transmission rate $\beta_t^{i 0}$ and the cooperative transmission rate $\beta_t^{i,\text{coop}}$ depend on a subset of the elements in the channel state matrix ${\bf H}_t$.

Let us first consider the direct $i \rightarrow \ell$ link with instantaneous channel
gain $h_t^{i\ell}$ and data rate $\beta_t^{i\ell} \in \mathcal{B}$ (bits/second) corrupted by
additive white Gaussian noise.
The bit error probability (BEP)  $P_t^{i \ell}(h_t^{i\ell},\beta_t^{i\ell})$ at the output of the maximum likelihood (ML)
detector of node $\ell$, under the assumption that a Gray code
is used to map the information bits into QAM symbols and the
signal-to-noise ratio (SNR) is sufficiently high, can be upper bounded as  (see \cite{Pro})
\be
P_t^{i \ell}(h_t^{i\ell},\beta_t^{i\ell}) \le
4 \, \exp{\left[- \frac{3 \, \gamma \, |h_t^{i\ell}|^2}{2 \left (2^{\beta_t^{i\ell}T_s}-1\right)}\right]} \: ,
\label{bound-direct}
\ee
where $\gamma \eqdef \frac{\mathcal{E}_s}{N_0}$ is
the average SNR per symbol expended by the transmitter and $N_0$ is the noise power spectral density.
Each direct transmission is subject to a PER threshold at the MAC sublayer, which
leads to a BEP constraint $P_t^{i \ell}(h_t^{i\ell},\beta_t^{i\ell}) \le BEP$
at the PHY layer. Consequently,  the achievable data rate $\beta_t^{i\ell}$ under the BEP constraint is
\be
\beta_t^{i\ell} = \frac{1}{T_s} \left \lfloor \log_2\left(1+\Gamma \, |h_t^{i\ell}|^2\right) \right \rfloor \: ,
\quad
\text{where}
\quad
\Gamma \eqdef  \frac{3 \, \gamma}{2 \, \left|\log_e \left(\frac{BEP}{4}\right)\right|} \: .
\label{beta}
\ee
The data rate $\beta_t^{i0}$ over the link between the source and the AP is obtained using \eqref{beta}
by setting $\ell=0$.
In this case, the number of symbols required to transmit a packet of $P$ bits is equal to 
$K_t^{i 0} \eqdef \lceil P/(\beta_t^{i 0}\, T_s) \rceil$. Thus, 
neglecting receive and processing energy consumption,
the energy required for a direct transmission of one packet is equal to
\be
\mathcal{E}_t^{i 0} \eqdef K_t^{i 0} \, \mathcal{E}_s =
\frac{P \, \mathcal{E}_s}{\beta_t^{i 0}\, T_s}  =
P \, \frac{\mathcal{P}_s}{\beta_t^{i 0}}  \: \text{ (Joules)}.
\label{E-0}
\ee
It is worth noting that the energy expended in direct mode is inversely proportional to the achievable data rate $\beta_t^{i 0}$.

At this point, let us consider the cooperative mode.
Because of possible error propagation, the end-to-end BEP
for a two-hop  cooperative transmission is 
cumbersome to calculate exactly with decode-and-forward relays;
therefore, the relationship that ties $\rate, \ratee$, and the relevant channel state information,
and that guarantees a certain reliability of the overall link, is not as simple as \eqref{beta}.
To significantly simplify the computation of $\rate$ and $\ratee$,
we use two different BEP thresholds $BEP_1$ and $BEP_2$
for the first and second hops, respectively.
The threshold $BEP_1$ is typically a large percentage of the total error rate budget, 
say $BEP_1 = 0.9 \, BEP$, and  $BEP_2 = BEP - BEP_1$, 
since the first link is the bottleneck
in decode-and-forward relaying. Indeed,  the performance at each relay is that of a single-input 
single-output system transmitting over a fading channel. On the other hand, 
the transmission over the second link (from the recruited relays to the destination)
can be regarded as a distributed multiple-input single-output system operating over 
a fading channel; consequently, the performance at the destination, which 
can take advantage from cooperative diversity, is significantly better 
than that of each source-to-relay link, even when a small number of
relays are recruited. Moreover, due to this fact and 
since the exponential function in \eqref{bound-direct}
decays fast as a function of its argument, 
we reasonably assume
that  the end-to-end BEP at the output of the ML detector of the AP is dominated
by the BEP over the worst source-to-relay channel, i.e., the link for which 
$|h_t^{i\ell}|$ is the smallest one.
Under this assumption, accounting for \eqref{beta}, we can estimate $\rate$ in Phase I as
\be
\rate =  \frac{1}{T_s} \left \lfloor \log_2\left(1+\Gamma_1 \, \min_{\ell \in {\mathcal C}_t^i} |h_t^{i\ell}|^2\right) \right \rfloor \:,
\label{beta_1}
\ee
where $\Gamma_1$ is obtained from $\Gamma$ by replacing $BEP$ with $BEP_1$.
In this phase, which lasts $R \, \rho_t^i \, x_t^i$ seconds, 
the number of  symbols needed to transmit a packet of $P$ bits is equal to $K_t^{i,1} = \lceil P/(\beta_t^{i,1}\, T_s) \rceil$ and, thus, it must result that 
\be
K_t^{i,1} \, T_s = \frac{P}{\beta_t^{i,1}}=R \, \rho_t^i \, x_t^i 
\quad \Longrightarrow \quad 
P = R \, \beta_t^{i,1} \rho_t^i \, x_t^i\: .
\label{cond-1}
\ee

Supposing that a subset ${\cal C}_t^i$ of the available nodes are recruited to serve as
relays in Phase II, these nodes, along with the $i$th user, cooperatively forward the source message
by using a randomized STBC rule \cite{Sirkeci07}. More specifically,  assuming error-free
demodulation at the decode-and-forward relays, if ${\bf a}_t^i \in \C^{K_t^{i,2}}$ gathers
the block of i.i.d. QAM source symbols to be transmitted in Phase II of time slot $t$, then
at the $\ell$th node, for each $\ell \in \{i\} \cup {\cal C}_t^i$, the vector ${\bf a}_t^i$ is mapped onto an
{\em orthogonal} space-time code matrix
$\bm{\mathcal{G}}({\bf a}_t^i) \in \C^{Q \times L}$ \cite{Tarokh}, where $Q$
is the block length and $L$ denotes the number of
antennas in the underlying space-time code.
During Phase II, the $\ell$th node transmits
a linear weighted combination of the columns of $\bm{\mathcal{G}}({\bf a}_t^i)$, with
the weights of the $L$ columns of $\bm{\mathcal{G}}({\bf a}_t^i)$
contained in the vector ${\bf r}_{\ell} \in \C^L$. We denote with
$\mathbf{R} \eqdef ({\bf r}_{\ell} \, | \, \ell \in {\cal C}_t^i) \in \C^{L \times N_t^i}$
the weight matrix of all the cooperating nodes, where $N_t^i \le M$  is the cardinality
of ${\cal C}_t^i$.\sfootnote{One specific code of the STBC matrix
is always assigned to the source itself, which transmits over the cooperative link every time cooperation is activated.
This can be accounted for by simply setting ${\bf r}_{i} = (1,0\ldots,0)^T$ and
replacing the first row of
$\mathbf{R}$ with $(0\ldots,0)$, whereas the remaining entries of
$\mathbf{R}$ are identically and independently generated random
variables with zero mean and variance $1/L$.}
Under the randomized STBC rule, the AP observes the space-time coded signal $\bm{\mathcal{G}}({\bf a}_t^i)$
with {\em equivalent} channel vector $\widetilde{\h}_t^{i,2} \eqdef {h}_t^{i 0} \, {\bf r}_{i} + \mathbf{R} \, {\h}_t^{i,2}$, where
${\h}_t^{i,2} \eqdef ({h}_t^{\ell 0} \, | \, \ell \in {\cal C}_t^i)^T \in \C^{N_t^i}$
collects all the channel coefficients between the
relay nodes and the AP (see Fig.~\ref{fig:network}).
Note that the AP only needs to estimate $\widetilde{\h}_t^{i,2}$
for coherent ML decoding and that the randomized coding is
decentralized since the $\ell$th relay chooses ${\mathbf r}_\ell$ locally.
By capitalizing on the orthogonality of the underlying STBC matrix $\bm{\mathcal{G}}({\bf a}_t^i)$,
the BEP $P_t^{i,2}(\widetilde{\h}_t^{i,2} ,{\beta}_t^{i,2})$ over the second hop at the output of the
ML detector of the AP using data rate  ${\beta}_t^{i,2}$ (bits/second) can be upper bounded as
in \eqref{bound-direct} by replacing  $|h_t^{i\ell}|^2$ and $\beta_t^{i\ell}$ with
$\|\widetilde{\h}_t^{i,2} \|^2$ and  ${\beta}_t^{i,2}$, respectively.
By imposing the BEP constraint
$P_t^{i,2}(\widetilde{\h}_t^{i,2} ,{\beta}_t^{i,2}) \le BEP_2$,
the data rate ${\beta}_t^{i,2}$ attainable on the second hop
of the cooperating link is given by
\be
{\beta}_t^{i,2}=\frac{1}{T_s} \left \lfloor \log_2[1+\Gamma_2 \, (|{h}_t^{i 0}|^2+\|\mathbf{R} \, {\h}_t^{i,2}\|^2)] \right \rfloor \: ,
\label{beta_2}
\ee
where $\Gamma_2$ is obtained from $\Gamma$ in \eqref{beta} by replacing $BEP$ with $BEP_2$.
In this phase, which lasts $R \, (1-\rho_t^i) \, x_t^i$ seconds, the number of  symbols needed to 
transmit a packet of $P$ bits is equal to 
$K_t^{i,2} = \lceil P/(\beta_t^{i,2}\, T_s) \rceil$ and, thus, it must result that 
\be
Q \, T_s =\frac{P}{R_{c} \, \beta_t^{i,2}}=R \, (1-\rho_t^i) \, x_t^i 
\quad \Longrightarrow \quad P = R \, R_{c} \, \beta_t^{i,2} \, (1-\rho_t^i) \, x_t^i  \: ,
\label{cond-2}
\ee
where $R_c \eqdef K_t^{i,2}/Q \le 1$ is the rate of the orthogonal STBC rule.
From \eqref{cond-1} and \eqref{cond-2}, the transmission time for the two phase communication mode is
\be
R \, x_t^i = \frac{P}{\beta_t^{i,1}} + \frac{P}{R_{c} \, \beta_t^{i,2}}
= P \underbrace{\left(\frac{1}{\beta_t^{i,1}} + \frac{1}{R_c \, \beta_t^{i,2}}\right)}_{\frac{1}{\beta_t^{i,\text{coop}}}} 
=\frac{P}{\beta_t^{i,\text{coop}}} \: ,
\label{beta-formula}
\ee
which also unveils what is the functional dependence of $\beta_t^{i,\text{coop}}$ on 
$\rate$ and $\ratee$. Moreover, from \eqref{cond-1} and \eqref{cond-2}, it is required that
\be
R \, \beta_t^{i,1} \rho_t^i \, x_t^i =R \, R_{c} \, \beta_t^{i,2} \, (1-\rho_t^i) \, x_t^i
\quad \Longrightarrow \quad
\rho_t^i=\frac{1}{1+ \rate/(\ratee \, R_c)} \: ,
\ee
which  shows that, given the STBC rule, the time fraction $\rho_t^i$ is
determined by the data rates in Phase I and II. The cooperative mode is activated 
only if the cooperative transmission is more data-rate efficient
than the direct communication, i.e., only if $\beta_t^{i,\text{coop}} >
\beta_t^{i 0}$,  which from \eqref{beta-formula} leads to the following condition
\be
\frac{1}{\rate}+\frac{1}{R_c \, \ratee} < \frac{1}{\beta_t^{i 0}}\: .
\label{coop-useful}
\ee
If condition \eqref{coop-useful} is fulfilled, then the opportunistically optimal 
cooperation decision is  $z_t^i = 1$ ;
otherwise, the $i$th source transmits to the AP in direct mode
and $z_t^i=0$.

It is interesting to evaluate the energy consumption in the case of a cooperative transmission.
Neglecting receive and processing energy consumption, 
the energy expended by the source $i$ for transmission of one packet is equal to
\be
\mathcal{E}_t^{i,\text{source}}= \left(K_t^{i,1}+K_t^{i,2}\right) \mathcal{E}_s =
P \, \frac{\mathcal{P}_s}{\beta_t^{i,\text{coop}}}  \: \text{ (Joules)},
\label{E-source}
\ee
whereas the energy expended by each recruited relay node for transmitting one packet of the $i$th source is given by
\be
\mathcal{E}_t^{i,\text{relay}}= K_t^{i,2} \, \mathcal{E}_s =
P \, \frac{\mathcal{P}_s}{\ratee \, R_c}  \: \text{ (Joules)}.
\label{E-relay}
\ee
It is noteworthy from \eqref{E-0} and \eqref{E-source}
that, since cooperation is activated only when 
$\beta_t^{i,\text{coop}} > \beta_t^{i 0}$, the
energy expended by the source node $i$ 
for a cooperative transmission is smaller than that  
required by the same node for a direct transmission.
On the other hand, the energy \eqref{E-relay} expended by the relays 
is inversely proportional to the achievable data rate in Phase II. 
Therefore,  provided that $\ratee \, R_c \gg  \beta_t^{i 0}$, over a sufficiently long period, 
the energy expenditure in relaying another node's data can be partially compensated for 
when the recruited relay acts as a source in the network.
The total energy expended in the network to transmit $\|{\bf y}_t^i\|_1$ packets for user $i$ can be expressed as
\be
\label{eq:total-energy}
\mathcal{E}_t^i \left({\bf y}_t^i, z_t^i, \mathcal{C}_t^i\right) = \left\{
\begin{array}{l l}
  \|{\bf y}_t^i\|_1 \, \mathcal{E}_t^{i0} \:, & \quad \mbox{if $z_t^i=0$} \:;\\
  \|{\bf y}_t^i\|_1 \left( \mathcal{E}_t^{i,\text{source}} + N_t^i \, \mathcal{E}_t^{i,\text{relay}} \right) \:, & \quad \mbox{if $z_t^i=1$} \: .\\
 \end{array} \right.
\ee
The energy consumption in the direct and cooperative modes is 
numerically compared in Section~\ref{sec:results}.

\vspace{-5mm}
\section{Cooperative Multi-User Video Transmission}
\label{sec:opt-scheduling}

Recall that ${\cal T}_t^i$ denotes the $i$th user's traffic state
and ${\bf H}_t$ collects the channel coefficients among all the nodes and the AP. Hence, the global state can be defined as ${\s}_t \eqdef \left({\cal T}_t^1, {\cal T}_t^2, \ldots, {\cal T}_t^M, {\bf H}_t\right) \in {\cal S}$, where ${\cal S}$ is a discrete set
of all possible states.\sfootnote{To have a discrete set of network states, the individual link states in ${\bf H}_t$ are quantized into a finite number of bins (see \cite{Wang04} for details).}
Since: (i) the $i$th user's traffic state evolves as a Markov process controlled by its scheduling action ${\bf y}_t^i$; (ii) the $i$th user's traffic state transition is conditionally independent of the other users' traffic state transitions given ${\bf y}_t^i$; and (iii) the state of each $i \rightarrow \ell$ link $h_t^{i \ell}$ is assumed to be i.i.d. with respect to time; the sequence of global states
$\{ \s_t : t \in \N\}$ can be modeled as a controlled Markov process with transition probability function
\be
p({\s}_{t+1} \, | \, {\s}_t, {\bf y}_t)
=p\left({\bf H}_{t+1} \right) \prod_{i=1}^M \, p({\cal T}_{t+1}^i \, | \, {\cal T}_t^i, {\bf y}_t^i) \: ,
\ee
where ${\bf y}_t \eqdef (\{{\bf y}_t^1\}^T, \{{\bf y}_t^2\}^T, \ldots, \{{\bf y}_t^M\}^T)^T$ collects the scheduling actions
of all the video users. 

Under the scheduling action ${\bf y}_t^i$, the $i$th user obtains the immediate utility
\be
u^i({\cal T}_t^i, {\bf y}_t^i) \eqdef \sum_{j \in {\cal F}_t^i} q_j^i \, y_{t,j}^i \:,
\ee
which is the total video quality improvement experienced by the $i$th user by taking scheduling action ${\bf y}_t^i$ in traffic state ${\cal T}_t^i$ under the assumption that quality is incrementally additive \cite{Chou06}.

The objective of the MU optimization is the maximization of the {\em expected discounted sum of utilities} with respect to the joint scheduling action ${\bf y}_t$ and the cooperation decision vector ${\bf z}_t \eqdef (z_t^1, z_t^2, \ldots, z_t^M)^T$ taken in each state ${\s}_t$.
Due to the stationary Markovian transition probability function, the optimization can be formulated as an MDP that satisfies the following dynamic programming equation\sfootnote{In this section, since we model the problem as a stationary MDP, we omit the time index when it does not create confusion. In place of the time index, we use the notation $(\cdot)^\prime$ to denote a state variable in the next time step (e.g. ${\cal T}^{i \prime}$, ${\bf H}^\prime$, ${\bf s}^\prime$).}
\be
U^*({\bf s})=\max_{{\bf y}, {\bf z}}
\left\{ \sum_{i=1}^M u^i({\cal T}^i,{\bf y}^i)
+\alpha \sum_{{\bf s}^\prime \in {\cal S}}
p({\bf H}^\prime) \prod_{i=1}^{M}p({\cal T}^{i \prime} \, | \, {\cal T}^i,{\bf y}^i) \, U^*({\bf s}^\prime) \right\},
\forall {\bf s},
\label{eq:mup-src-coop}
\ee
subject to
\be
{\bf y}^i \in {\cal P}^i({\cal T}^i,{\beta^i}) \quad \text{and} \quad \sum_{i=1}^{M} x^i \leq 1
\ee
where $x^i$ is the time-fraction allocated to the $i$th user given its scheduling action ${\bf y}^i$ and transmission rate $\beta^i$, i.e.,
\be
\label{eq:sched2rsrc}
x^i = \frac{P}{R \, \beta^i} \, \|{\bf y}^i\|_1 \: ,
\ee
the parameter $\alpha \in [0,1)$ is the ``discount factor'', which accounts for the relative importance of the present and future utility, and ${\cal P}^i({\cal T}^i,{\bf H})$ is the set of feasible scheduling actions given the traffic state ${\cal T}^i$ and channel state matrix ${\bf H}$. From Theorem 6.2.5 in~\cite{PutermanMDP}, we know that there exists a stationary optimal policy that is the global optimal solution to~\eqref{eq:mup-src-coop} .

Given the distributions $p({\bf H})$ and $p({\cal T}^{i \prime} \, | \, {\cal T}^i,{\bf y}^i)$ for all $i$, the above MU-MDP can be solved by the AP using value iteration or policy iteration \cite{Bertsekas05}.
However, there are two challenges associated with solving the above MU-MDP.
First, the complexity of solving an MDP is proportional to the cardinality of its state-space ${\cal S}$, which, in the above MU-MDP, scales exponentially with the number of users, i.e., $M$, and with the number of links in ${\bf H}$, i.e., $M^2$. Hence, even for moderate sized networks, it is impractical to compute, or even to encode, $U^*({\bf s})$. In subsection~\ref{sec:simplify}, we show that the exponential dependence on the number of links in ${\bf H}$ can be eliminated.
Second, in the uplink scenario, the traffic state information is local to the users, so neither the AP nor the users have enough information to solve the above MU-MDP. In subsection~\ref{sec:distributed}, we summarize the findings in~\cite{Fu09a} that show that the considered optimization can be approximated to make it amenable to a distributed solution. Additionally, this distributed solution eliminates the exponential dependence on the number of users. Note that the simplification in subsection~\ref{sec:simplify} is very important, because only after obtaining this result does it become possible to use the solution in~\cite{Fu09a}.

\vspace{-5mm}
\subsection{Reformulation with simplified network state}
\label{sec:simplify}

The only reason to include the detailed network state information ${\bf H}$ and the cooperation decision ${\bf z}$ in the MU-MDP is to make foresighted cooperation decisions, which take into account the impact of the immediate cooperation decision on the expected future utility of the users. However, if we can show that the optimal opportunistic (i.e., myopic) cooperation decision is also long-term optimal, then the detailed network state information does not need to be included in the MU-MDP. The following theorem shows that the optimal opportunistic cooperation decision, which maximizes the immediate transmission rate, is also long-term optimal.

\begin{theorem}[Opportunistic cooperation is optimal]
\label{Thm-1}
If utilizing cooperation incurs zero cost to the source and relays, then the optimal opportunistic cooperation decision, which maximizes the immediate throughput, is also long-term optimal.
\end{theorem}

\begin{proof}
See Appendix~I.
\end{proof}

To intuitively understand why maximizing the immediate transmission rate at the PHY layer is long-term optimal, consider what happens when a user chooses not to maximize its immediate transmission rate (i.e., does not utilize the optimal opportunistic cooperation decision). Two things can happen: either less packets are transmitted overall because of packet expirations; or, the same number of packets are transmitted overall, but their transmission incurs additional resource costs because transmitting the same number of packets at a lower rate requires more resources [see \eqref{eq:sched2rsrc}]. In either case, the long-term utility is suboptimal.
A consequence of Theorem~\ref{Thm-1} is that the cooperation decision vector ${\bf z}$ does not need to be included in the MU-MDP. Instead, it can be determined opportunistically by selecting ${\bf z}$ to maximize the immediate transmission rate. Most importantly, this means that the MU-MDP does not need to include the high-dimensional network state.

We now make two remarks regarding Theorem~\ref{Thm-1} so that its consequences are not misinterpreted.
First, in the introduction, we noted that maximizing throughput is a suboptimal multiple access strategy for wireless video. This does not contradict Theorem~\ref{Thm-1} because it only states that the {\em cooperation decision} should be made opportunistically to maximize the immediate transmission rate. Indeed, myopic (opportunistic) resource allocation and scheduling is suboptimal because it does not take into account the dynamic video data attributes (i.e., deadlines, priorities, and dependencies).
Second, although the users' MDPs do not need to include the high-dimensional network state, the optimal resource allocation and scheduling strategies still depend on it; however, instead of tracking ${\bf H}_t$, it is sufficient to track the users' optimal opportunistic transmission rates provided by the PHY layer, i.e., $\beta_t^i$ for all $i$. Under the assumption that the channel coefficients are i.i.d. random variables with respect to $t$, $\beta_t^i$ can also be modeled as an i.i.d. random variable with respect to $t$. We let $p(\beta^i)$ denote the probability mass function (pmf) from which $\beta_t^i$ is drawn. We note that $p(\beta^i)$ depends on $p({\bf H})$ and the deployed PHY layer cooperation algorithm.

Based on the second remark, we can simplify the maximization problem in \eqref{eq:mup-src-coop}. Let us define the $i$th user's state as $s^i \eqdef \left({\cal T}^i,\beta^i \right) \in {\cal S}^i$ and redefine the global state as ${\bf s} \eqdef (s^1,\ldots,s^M)^T$. In Section~\ref{protocol}, we describe how $\beta^i$ is determined, but for now we will take for granted that it is known. Because the optimization does not need to include the cooperation decision, the maximization of the expected sum of discounted utilities in \eqref{eq:mup-src-coop} can be simplified by only maximizing with respect to the scheduling action ${\bf y}$ in each state ${\bf s}$, that is,
\be
U^*({\bf s})=\max_{{\bf y}}
\left\{ \sum_{i=1}^M u^i({\cal T}^i,{\bf y}^i)
+\alpha \sum_{{\bf s}^\prime \in {\cal S}}
\prod_{i=1}^{M}p({s}^{i \prime} \, | \, s^i, {\bf y}^i) \, U^*({\bf s}^\prime) \right\},
\forall {\bf s},
\label{eq:mup-src}
\ee
subject to
\be
{\bf y}^i \in {\cal P}^i({\cal T}^i,{\beta}^i) \quad \text{and} \quad
\displaystyle \sum_{i=1}^{M} x^i \leq 1,
\ee
where $p({s}^{i \prime} \, | \, s^i, {\bf y}^i) = p({\beta^{i \prime}}) \, p({\cal T}^{i \prime} \, | \, {\cal T}^i,{\bf y}^i)$.

\vspace{-5mm}
\subsection{Distributed solution}
\label{sec:distributed}

Similar to~\cite{Fu09a}, \eqref{eq:mup-src} can be reformulated as an unconstrained MDP using Lagrangian relaxation. The key idea is to introduce a Lagrange multiplier $\lambda_{\bf s}$ associated with the stage resource constraint $\sum_{i=1}^{M} x^i \leq 1$ in each global state ${\bf s}$ because every global state has a different resource-quality tradeoff. The resulting dual solution has zero duality gap compared to the primary problem [i.e., \eqref{eq:mup-src}], but it still depends on the global state so it is not amenable to a distributed solution. However, by imposing a uniform resource price $\lambda_{\bf s}=\lambda$, $\forall {\bf s} \in {\cal S}$, which is independent of the multi-user state, the resulting MU-MDP can be decomposed into $M$ MDPs, one for each user~\cite{Fu09a}.\sfootnote{We note that the resource price is only used to efficiently allocate the limited wireless resources among the users; it is not used to generate revenue for the AP. In other words, it is a congestion price rather than a real price.} These local MDPs satisfy the following dynamic programming equation
\be
U^{i, *} (s^i, \lambda) = \max_{{\bf y}^i}
\left[u^i({\cal T}^i,{\bf y}^i)- \lambda\left(x^i-\frac{1}{M}\right) +
\alpha \sum_{{s}^{i \prime} \in \mathcal{S}} p({s}^{i \prime} \, | \, s^i, {\bf y}^i) \, U^{i,*}({s}^{i \prime}, \lambda)\right] \: ,
\label{eq:user_bellman}
\ee
\be
\hat{U}^{\lambda^*}({\bf s}) = \min_{\lambda \geq 0}
\sum_{i=1}^M U^{i,*}(s^i,\lambda) \: ,
\label{eq:sum_user_bellman}
\ee
subject to ${\bf y}^i \in {\cal P}^i({\cal T}^i,\beta^i)$. Importantly, the $i$th user's dynamic programming equation defines the optimal scheduling action as a function of the $i$th user's state, rather than the global state ${\bf s}$. In this paper, the $i$th user solves \eqref{eq:user_bellman} offline using value iteration; however, it can be easily solved online using reinforcement learning as in \cite{Fu09a} and \cite{Salodkar10}. Also, note that due to the distributed nature of the proposed algorithm, the stage resource constraint $\sum_{i=1}^M x_t^i  \leq 1$ is not guaranteed to be satisfied during convergence or at steady-state. Because the stage resource constraint may be violated, it must be enforced separately by the AP, which we assume normalizes the requested resource allocations and, subsequently, has the users recompute their scheduling policies to satisfy the new allocations.

Although the optimization can be decomposed across the users, the optimal resource price $\lambda$ still depends on all of the users' resource demands. Hence, $\lambda$ must be determined by the AP in both the uplink and downlink scenarios. Specifically, the resource price can be numerically computed by the AP using the subgradient method.
The subgradient with respect to $\lambda$ is given by
$\sum_{i=1}^{M}X^i-\frac{1}{1-\alpha}$, where $X^i = E\left[\sum_{t=0}^{+\infty}\alpha^t x_t^i \, | \, s^i_0 \right]$ is the $i$th user's expected discounted accumulated resource consumption, which can be calculated as described in~\cite{Fu09a}. Importantly, $X^i$ can be computed locally by the $i$th user in the uplink scenario and by the AP in the downlink scenario. Using the subgradient method, the resource price is updated as
\be
\lambda^{k+1}=\left[\lambda^k + \mu^k \left(\sum_{i=1}^{M}X^i - \frac{1}{1-\alpha}\right)\right]^+ \:,
\ee
where $\mu^k$ is a diminishing step size.
Since the focus of this paper is on the interaction between the multiuser video transmission and the cooperative PHY layer, we refer the interested reader to~\cite{Fu09a} for complete details on the dual decomposition outlined in this subsection, and the derivation of the subgradient with respect to $\lambda$.

We note that a similar decomposition has recently been proposed for energy-efficient uplink scheduling with delay constraints in multiuser wireless networks using a different MU-MDP framework \cite{Salodkar10}. Besides the fact that \cite{Salodkar10} does not consider physical layer cooperation or heterogeneous traffic characteristics, there is one significant difference between the decomposition in \cite{Salodkar10} and the one adopted in this paper. Specifically, the TDMA-like protocol in \cite{Salodkar10} assumes that only one user can transmit in each time slot, whereas we consider a TDMA-like protocol in which each time slot is divided into different length transmission opportunities for each user. Moreover, in \cite{Salodkar10}, every user has a unique Lagrange multiplier associated with its average buffer delay constraint. In contrast, in our decomposition, all users have the same Lagrange multiplier, which regulates the resource division among the users, rather than their individual delay constraints. Note that, in this paper, delay constraints are included in the application model. Importantly, Theorem 1 applies to the MU-MDP formulation in \cite{Salodkar10} and therefore  the recruitment protocol proposed in Section~\ref{protocol} can be used to integrate cooperation into \cite{Salodkar10}. In other words, the novelty and technical contributions of this paper are independent of the dual decomposition in~\cite{Fu09a}, which we only use for illustrative purposes.

\vspace{-5mm}
\section{Recruitment protocol}
\label{protocol}
With reference to the uplink scenario, we define our opportunistic cooperative 
strategy to select distributively the set of cooperative relays ${\cal C}_t^i$
and make the decision $z_t^i$ at the AP. The downlink case is a minor variation.

Importantly, the AP can exactly evaluate ${\beta}_t^{i,2}$ in \eqref{beta_2}
because it can estimate ${h}_t^{i 0}$ and $\mathbf{R} \, {\h}_t^{i,2}$ via training as mentioned in
Section~\ref{sec:coopPHY}. However, the trouble in recruiting relays on-the-fly is that the AP and the relays cannot directly
compute $\rate$ given by \eqref{beta_1}, since they cannot estimate the 
channel coefficients $h_t^{i \ell}$, for all $\ell \in {\mathcal C}_t^i$.
Some MAC randomized protocols have
recently been proposed \cite{Verde,Liu}, which get around the problem that
the AP and the relays do not have the necessary channel state information
to determine $\rate$. However, such protocols require the exchange
and/or the tracking of a large amount of network parameters that may incur unacceptable delays 
in a wireless video network. In particular, the first- and second-hop data rates are computed 
in \cite{Liu} by the source node using the average PER evaluated by simulations.
To quickly setup the cooperative transmission and, thus, reduce 
the delays, we propose a much simpler recruitment scheme that is based on
the closed-form formulas \eqref{beta_1} and \eqref{beta_2}. 
The proposed four-way protocol is reminiscent of the request-to-send (RTS) and clear-to-send (CTS)
handshaking used in carrier sense multiple access with collision avoidance (CSMA/CA),
which is extended to include a helper-ready to send (HTS) control message that is 
cooperatively transmitted by the relays using randomized STBC and 
a cooperative recruitment signal (CRS) that is sent by the AP to recruit relays. 
The idea of sending the HTS frame in cooperative mode has been originally 
proposed in \cite{Liu}. However, 
apart from the use of the HTS control message, the proposed protocol is different from that of \cite{Liu}
because we use a completely different recruitment policy.

All the control frames are transmitted at the base rate $\beta_0$ 
such that they can be decoded correctly, and 
the thresholds $BEP_{1}$ and $BEP_{2}$, as well as $L$ and $R_c$,
are fixed parameters that are known at all the nodes.
Fig.~\ref{fig:signaling} illustrates the signaling protocol for time slot $t$,
which consists of the nine steps detailed in Table~\ref{tab:protocol}.
We would like to highlight that, similar to the data transmitted in 
Phase II, the HTS message is a cooperative signal, i.e., all relays
jointly deliver the HTS frame using randomized STBC at the same time
and, hence, simultaneous transmissions do not cause a collision.
With reference to Table~\ref{tab:protocol},
the key observation is that the selection of the set ${\mathcal C}_t^i$
by virtue of \eqref{eq.Ci} is done in a distributed way and, moreover, by simply having
access to the channel state from the source $i$ to itself, i.e.,
$h_t^{i \ell}$, the $\ell$th candidate cooperative node can
{\em autonomously} determine if, by cooperating, it can improve the data rate of node $i$.
Another important observation is that the recruitment of the cooperative nodes and the assignment of the
data rates requires only four control messages for each source. In particular, the control information
exchange is independent of the number of recruited relays thanks to the randomization of the
cooperative transmission.
Moreover, the two parameters $\xi_t$ and $L$ need to be chosen appropriately.
The best choice for $\xi_t$ and $L$ requires global network information. A learning framework
would be very appropriate for their selection but we defer the treatment of this aspect to future work.
Finally, as for the impact of $L$ on the network performance, it should evidenced that randomized channels tend to behave statistically like their non-randomized counterparts \cite{Sirkeci07}, with deep-fade events that become as frequent as those of $L$ independent channels, as long as the number of cooperative nodes $N_t^i \geq L+1$.

\vspace{-5mm}
\section{Numerical Results}
\label{sec:results}

We consider a network with 50 potential relay nodes placed randomly and uniformly throughout the 100 m coverage range of a single AP as illustrated in Fig.~\ref{fig:snapshot}. We specify the placement of the video source(s) separately for each experiment.
Let $\eta_t^{i\ell}$ denote the distance in meters between the $i$th and $\ell$th nodes. The fading coefficient $h_t^{i\ell}$ over the $i \rightarrow \ell$ link is modeled as an i.i.d. ${\cal CN}(0, (\eta_t^{i\ell})^{-\delta})$ random variable, where $\delta$ is the path-loss exponent. Additionally, we assume that the entries of $\bR$, defined in Section~\ref{sec:coopPHY}, are i.i.d. ${\cal CN}(0, \frac{1}{L})$ random variables, where $L$ is the length of the STBC. If an error occurs in the packet transmission, then the packet remains in the frame buffer to be retransmitted in a future time slot (assuming the packet's deadline has not passed).

Due to space constraints, and because cooperation has the same impact in both uplink and downlink scenarios, we only present results for cooperative uplink video transmission. In particular, we consider four uplink scenarios:
\begin{enumerate}

\itemsep=0mm

\item
{\bf Single source}: In this scenario, we assume that a single source node is placed between 10 and 100 m directly to the right of the AP in Fig.~\ref{fig:snapshot}. We use this scenario to evaluate the transmission rates in the direct and cooperative transmission modes at different distances from the AP, and to determine a good self-selection parameter $\xi$.

\item
{\bf Homogeneous video sources}: This scenario mimics a surveillance application in which three cameras capture correlated video content in an outdoor environment and transmit it to the AP. The video sources are placed to the right of the AP as illustrated in Fig.~\ref{fig:relay_freqs}. To simulate correlated content, we assume that each of the three cameras stream the Foreman sequence (CIF resolution, 30 Hz framerate, encoded at 1.5 Mb/s) offset by several frames. Using homogeneous sources allows us to isolate the impact of cooperation on the video streaming performance by removing the additional layer of complexity introduced by heterogeneous video sources (e.g. different packet priorities and bit-rates among the video users).

\item
{\bf Heterogeneous video sources 1}: This scenario mimics a network in which users deploy entertainment applications such as video sharing or video conferencing. To simulate this, we assume that the three video sources illustrated in Fig.~\ref{fig:relay_freqs} transmit heterogeneous video content to the AP. Specifically, we assume that video user 1 streams the Coastguard sequence (CIF, 30 Hz, 1.5 Mb/s), video user 2 streams the Mobile sequence (CIF, 30 Hz, 2.0 Mb/s), and video user 3 streams the Foreman sequence (CIF, 30 Hz, 1.5 Mb/s).

\item
{\bf Heterogeneous video sources 2}: This is the same as the previous scenario, but with video user 2 streaming the Foreman sequence and video user 3 streaming the Mobile sequence.

\end{enumerate}
We note that the proposed framework can be applied using any video coder to compress the video data. However, for illustration, we use a scalable video coding scheme~\cite{ohm94}, which is attractive for wireless streaming applications because it provides on-the-fly application adaptation to channel conditions, support for a variety of wireless receivers with different resource and power constraints, and easy prioritization of video packets.

In our results, we deploy the proposed randomized STBC cooperation protocol outlined in Table~\ref{tab:protocol} and determine the optimal resource allocation and scheduling decisions using the distributed optimization introduced in Section~\ref{sec:distributed}. The relevant simulation parameters are given in Table~\ref{table:parameter_table}. Note that, in the homogeneous and heterogeneous scenarios described above, we simulate a network with a ``high'' transmission rate, using the symbol rate $\frac{1}{T_s} = 1250000$, and a network with a ``low'' transmission rate, using the symbol rate $\frac{1}{T_s} = 625000$ symbols/second.

\vspace{-5mm}
\subsection{Transmission rates and energy consumption}
\label{sec:coop-details}

In this subsection, we consider the single source scenario described above. Fig.~\ref{fig:coop_stats} illustrates the performance of the proposed cooperation protocol for time-invariant self-selection parameter values $\xi_t = \xi \in \{0.1,0.2,\ldots,0.5\}$, and the performance of direct transmission, given a single source transmitting to the AP. Note that these results hold regardless of the symbol rate. In particular, the ``transmission rate'' in Fig.~\ref{fig:coop_stats}(a) is presented in terms of the spectral efficiency (bits/second/Hz); the probability of cooperation in Fig.~\ref{fig:coop_stats}(b) and the average number of recruited relays in Fig.~\ref{fig:coop_stats}(c) only depend on the spectral efficiency; and the energy results reported in Figs.~\ref{fig:coop_stats}(d-f) are normalized by setting the symbol energy ${\cal E}_s =\frac{T_s}{P} $ (or, equivalently, ${\cal P}_s = \frac{1}{P}$) in \eqref{E-0}, \eqref{E-source}, and \eqref{E-relay}.

From Fig.~\ref{fig:coop_stats}(a), it is clear that nodes further from the AP utilize cooperation more frequently than nodes closer to the AP. This is because, on average, distant nodes have the feeblest direct signals to the AP due to path-loss and, therefore, have the most to gain from the channel diversity afforded to them by cooperation. It is also clear from Fig.~\ref{fig:coop_stats}(a) that cooperation is utilized more frequently as the self-selection parameter $\xi$ increases. This is because, as illustrated in Fig.~\ref{fig:coop_stats}(c), more relays satisfy the self-selection condition in step 5 of Table~\ref{tab:protocol} for larger values of $\xi$. However, larger values of $\xi$ yield relay nodes for which $\frac{\beta_t^{i 0}}{\beta_t^{i \ell}}$ is large, which leads to a bad transmission rate over the bottleneck hop-1 cooperative link. Due to this poor bottleneck rate and the large number of recruited relays, the average transmission rate shown in Fig.~\ref{fig:coop_stats}(b) declines for $\xi > 0.2$ even while the total energy consumption increases as illustrated in Fig.~\ref{fig:coop_stats}(d). In contrast, lower values of the self-selection parameter (e.g. $\xi < 0.2$) lead to too few nodes being recruited to achieve large cooperative gains, but yield lower energy consumption. Interestingly, the same properties of relay nodes that are desirable for achieving the best transmission rate -- a balance between the number and quality of relays -- is also important for achieving a high throughput-to-energy ratio. For example, Fig.~\ref{fig:coop_stats}(e) shows us that at 100 m from the AP, the average throughput-to-energy ratio for cooperative transmission with $\xi = 0.2$ is a little less than 0.8, which is close to the throughput-to-energy ratio of a direct transmission, which is 1 at 100 m.

Although the average network energy required to support a cooperative transmission is larger than that required for a direct transmission, this increase is moderate compared to the amount of energy the source node would have to expend in order to achieve the same transmission rate as the cooperative transmission, 
i.e., to attain $\beta_t^{i 0}= \beta_t^{i,\text{coop}}$ requires a large increase
in the transmission power with respect to the cooperative case. This is illustrated in Fig.~\ref{fig:coop_stats}(f), where, for example, it is shown that transmitting in the direct mode at the rate attainable under cooperative transmission with $\xi = 0.2$ requires approximately 13.5 normalized Joules/Packet compared to approximately 3.5 normalized Joules/Packet in the cooperative case shown in Fig.~\ref{fig:coop_stats}(d).\sfootnote{The results in Fig.~\ref{fig:coop_stats}(f) were obtained by fixing the transmission rate and adapting the symbol energy, which is in contrast to the current problem formulation in which we fix the symbol energy and adapt the transmission rate. Specifically, we calculated the symbol energy $\tilde{\cal E}_s$ required to set $\beta_t^{i \ell} = \beta_t^{i,\text{coop}}$ by rearranging \eqref{beta}. Note that we could also force $\beta_t^{i,\text{coop}} = \beta_t^{i 0}$ to achieve lower energy consumption at the same transmission rate as the direct mode.}

In the remainder of our experiments, we let the self-selection parameter $\xi_t=\xi=0.2$ because, as illustrated in Figs.~\ref{fig:coop_stats}(b,e), this value provides a large average transmission rate over the AP's entire coverage range and a high throughput-to-energy ratio.
With $\xi = 0.2$, Fig.~\ref{fig:relay_freqs} illustrates the activation frequencies for different relays and Fig.~\ref{fig:burden} illustrates the average energy consumed by the source and relay nodes. Notice that, under a cooperative transmission, the source node actually uses less power than under a direct transmission, which partially compensates for the extra energy it may expend acting as a relay for other nodes.

\vspace{-5mm}
\subsection{Transmission rate, resource price, and resource utilization}
\label{sec:rate-price}

Fig.~\ref{fig:tx_rates} illustrates the average transmission rates achieved by the video users in the homogeneous and heterogeneous scenarios in networks that support high and low transmission rates. Recall that the resource cost $x_t^i$ incurred by user $i$ is inversely proportional to the transmission rate [see~\eqref{eq:sched2rsrc}], which decreases as the distance to the AP increases due to path loss. Hence, when only direct transmission is available, user 3 tends to resign itself to a low average transmission rate because the cost of using resources is too high. Cooperation increases the average transmission rate, thereby providing user 3 lower cost access to the channel to transmit more data.

In the homogeneous scenario illustrated in Fig.~\ref{fig:tx_rates}(a), cooperation tends to equalize the resource allocations to the three users (this is especially evident in the cooperative case with a high transmission rate). This is because the homogeneous users have identical utility functions; thus, when sufficient resources are available, it is optimal for them to all operate at the same point of their resource-utility curves. In contrast, when heterogeneous users with different utility functions are introduced, the transmission rates change to reflect the priorities of the different users' video data. Observing Fig.~\ref{fig:tx_rates}(b,c), it is clear that the additional resources afforded by cooperation tend to go to the highest priority video user, who, in our simulations, is the user streaming the Mobile sequence.

Recall that users autonomously optimize their resource allocation and scheduling actions given the resource price $\lambda$ announced by the AP. Table~\ref{table:rsrc_price} illustrates the optimal resource prices in the homogeneous and heterogeneous scenarios along with the average network resource utilization, i.e. the average of $\sum_{i=1}^M x_t^i$. There are several interesting results in Table~\ref{table:rsrc_price}.
First, the average network resource utilization is often considerably less than the total available resources. This is due to the distributed nature of the resource allocation and scheduling algorithm, which requires users to be conservative in their resource usage to ensure feasible allocations.
Second, in the cooperative transmission mode, the resource price tends to increase and the utilization tends to decrease when going from a high rate to a low rate network, regardless of the streaming scenario. The resource price increases because the network supports lower rates, but the demand stays the same, which increases congestion. The utilization decreases because lower rates yield a coarser set of feasible resource allocations for each user (see \eqref{eq:sched2rsrc}).
Third, in the high rate network, the resource price tends to decrease and the utilization tends to increase when going from the direct to the cooperative transmission mode, regardless of the streaming scenario. The resource price decreases because cooperation floods the network with resources without significantly impacting demand, which reduces congestion. The utilization increases because the cooperative transmission mode supports higher transmission rates, which yield a finer set of feasible resource allocations for each user (see \eqref{eq:sched2rsrc}).
Finally, in the low rate network, the resource price and utilization tend to increase when going from the direct to the cooperative transmission mode. In contrast to the high rate network, the resource price {\em increases} because users that resigned themselves to very low transmission rates in the direct scenario suddenly demand resources when cooperation is enabled. The resource price increases in our simulations because the enlarged demand pool exceeds the additional supply of resources that is introduced by cooperation. In other words, users that would like to transmit video, but are too far from the AP for a direct transmission, are essentially absent from the network when only direct transmission is available, and therefore do not significantly impact the resource price and resource utilization; however, when cooperation is enabled, these users are suddenly within range of the AP, and will therefore demand resources, which increases congestion. As in the other cases, the utilization increases because the transmission rate increases.

\vspace{-5mm}
\subsection{Discounted utility and video quality comparison}
\label{sec:video-performance}

Table~\ref{table:discounted_utility_table} compares the expected value of the objective function in \eqref{eq:sum_user_bellman} (with respect to the stationary distribution over the states) obtained in the homogeneous and heterogeneous scenarios. Because the objective function includes a Lagrangian cost term, it is not always indicative of the corresponding video quality. For this reason, we also include  
Table~\ref{table:psnr_table} to compare the video quality obtained in the homogeneous and heterogeneous scenarios, where video quality is measured in terms of peak-signal-to-noise ratio (PSNR in dB) of the luminance channel.
In the network that supports a high transmission rate, the user furthest from the AP (user 3) benefits on the order of 5-10 dB PSNR from cooperation, while the video user closest to the AP (user 1) is penalized by less than 0.4 dB PSNR. In the network that only supports low transmission rates, user 3 goes from transmitting too little data to decode the video (denoted by ``$---$'') to transmitting enough data to decode at low quality, while penalizing user 1 by less than 0.8 dB PSNR.
Note that these PSNR results implicitly reflect the end-to-end delay from the source, through the relays, to the destination. This is because the sophisticated traffic model in subsection~\ref{sec:app-model} accounts for the fact that frames that are not entirely received before their deadlines, and frames that depend on them, cannot be decoded and therefore do not contribute to the received video quality.

\vspace{-5mm}
\section{Conclusion}
\label{sec:concl}

We introduced a cooperative multiple access strategy that enables nodes with high priority video data to be serviced while simultaneously exploiting the diversity of channel fading states in the network using a randomized STBC cooperation protocol.
We formulated the dynamic multi-user video transmission problem with cooperation as an MU-MDP and we used Lagrangian relaxation with a uniform resource price to decompose the MU-MDP into local MDPs at each user.
We analytically proved that opportunistic (myopic) cooperation strategies are optimal, and therefore the users' local MDPs only need to determine their optimal resource allocation and scheduling policies based on their experienced cooperative transmission rates.
Subsequently, we proposed a randomized STBC cooperation protocol that enables nodes to opportunistically and distributively self-select themselves as cooperative relays.
Finally, we experimentally showed that the proposed cooperation strategy significantly improves the video quality of nodes with feeble direct links to the AP, without significantly penalizing other users, and with only moderate increases in total network energy consumption.


\section*{Appendix~I: proof of Theorem~\ref{Thm-1}}
\label{A}

The transmission rate $\beta^i$ is a function of the cooperation decision $z^i$ and the channel state ${\bf H}$, i.e., we can write $\beta^i = \beta^i \left( {\bf H} ,z^i \right)$. Thus, the cooperation decision impacts the immediate utility because it constrains the set of feasible scheduling actions ${\cal P}^i \left( {\cal T}^i ,\beta^i \right)$ through the packet constraint $\|\mathbf{y}^i\|_1\leq \frac{R  \beta^i}{P}$.

Let $z_{{\rm{opp}}}^{i * }  = \arg {\max }_{z^i } \left\{ {\beta^i \left( {\bf H},z^i \right)} \right\}$ and $\beta_{{\rm{opp}}}^{i * }  = {\max}_{z^i} \left\{ \beta^i \left( {\bf H},z^i \right) \right\}$ denote the optimal opportunistic cooperation decision and the maximum transmission rate, respectively. Selecting the cooperation decision that maximizes the immediate transmission rate enlarges the set of feasible scheduling actions, i.e., ${\cal P}^i \left( {\cal T}^i ,\beta^i \right) \subseteq {\cal P}^i \left( {\cal T}^i ,\beta_{{\rm{opp}}}^{i * } \right)$, for all $\beta^i  \le \beta_{{\rm{opp}}}^{i * }$.
We now show that the optimal opportunistic cooperation decision enables a user to maximize its long-term utility for any $\alpha \geq 0$.
Let
$ u^i_\lambda \left( {{\cal T}^i,\beta^i, {\bf y}^i} \right) = \sum\nolimits_{j \in {\cal F}^i } {q_j^i y_j^i }  - \lambda \left( {x^i  - \frac{1}{M}} \right) $
denote the utility less the cost, where $x^i$ is given by \eqref{eq:sched2rsrc}. Under the optimal opportunistic cooperation decision, we have
\begin{eqnarray}
U_\lambda^{i,*}  \left( s^i \right)
&=& \max_{{\bf y}^i  \in
{\cal P}^i \left( {\cal T}^i ,\beta_{\rm{opp}}^{i *}\right)} \left\{ u^i_\lambda \left( {\cal T}^i, \beta_{\rm{opp}}^{i *}, {\bf y}^i\right)
 + \alpha \sum_{s^{i \prime}} {p\left( s^{i \prime}  |s^i ,{\bf y }^i  \right)U_\lambda^{i, *}  \left( s^{i \prime} \right)}  \right\} \\
&\ge& \max_{{\bf{y}}^i  \in {\cal P}^i \left( {{\cal T}^i,\beta^i} \right)} \left\{ u^i_\lambda \left( {\cal T}^i, \beta^i, {\bf y}^i \right) + \alpha \sum_{s^{i \prime}} {p\left( {s^{i \prime}  | s^i ,{\bf y}^i } \right) \bar{U}_\lambda^{i, *}  \left( {s^{i \prime}  } \right)}\right\}= \bar{U}_\lambda^{i}  \left( s^i \right) \:,
\end{eqnarray}
where the inequality is due to the fact that ${\cal P}^i \left( {\cal T}^i ,\beta^i \right) \subseteq {\cal P}^i \left( {\cal T}^i, \beta_{\rm{opp}}^{i * } \right)$ for all $\beta^i \leq \beta_{\rm{opp}}^{i * }$. Thus, the optimal opportunistic cooperative decision maximizes the long-term utility.


\bibliographystyle{IEEEbib}

\bibliography{strings,refs}

%
%

\newpage

\begin{figure}[t]
\centering
\includegraphics[width=0.8\linewidth]{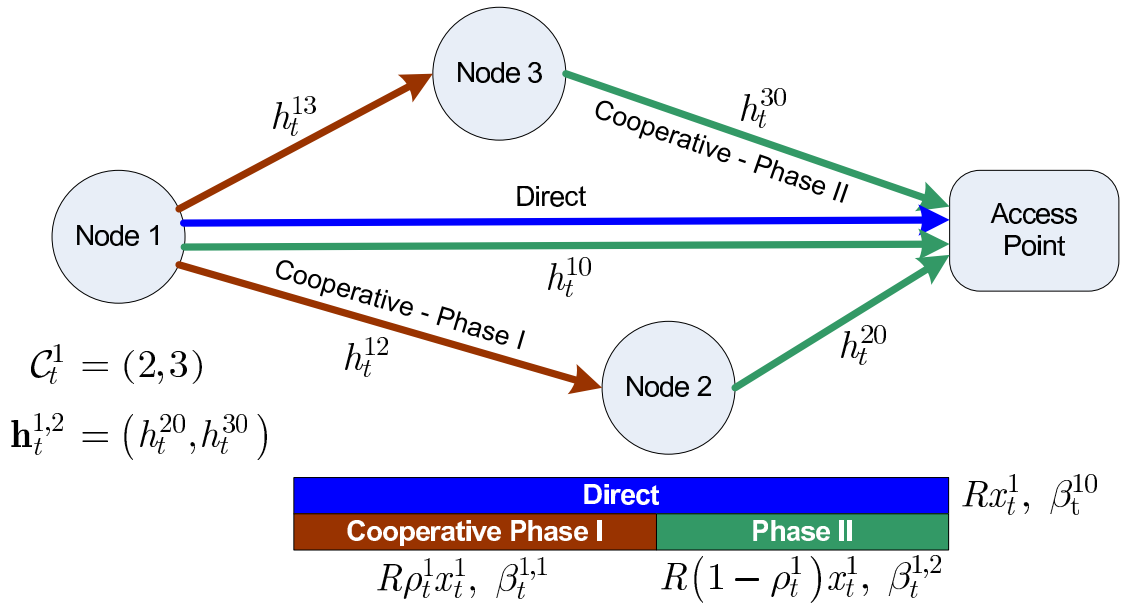}
\caption{An uplink wireless video network with cooperation. A downlink wireless video network with cooperation can be visualized by switching the positions of node 1 and the access point.}
\label{fig:network}
\end{figure}

%
%

\begin{figure}[t]
\centering
\includegraphics[width=1.0\linewidth]{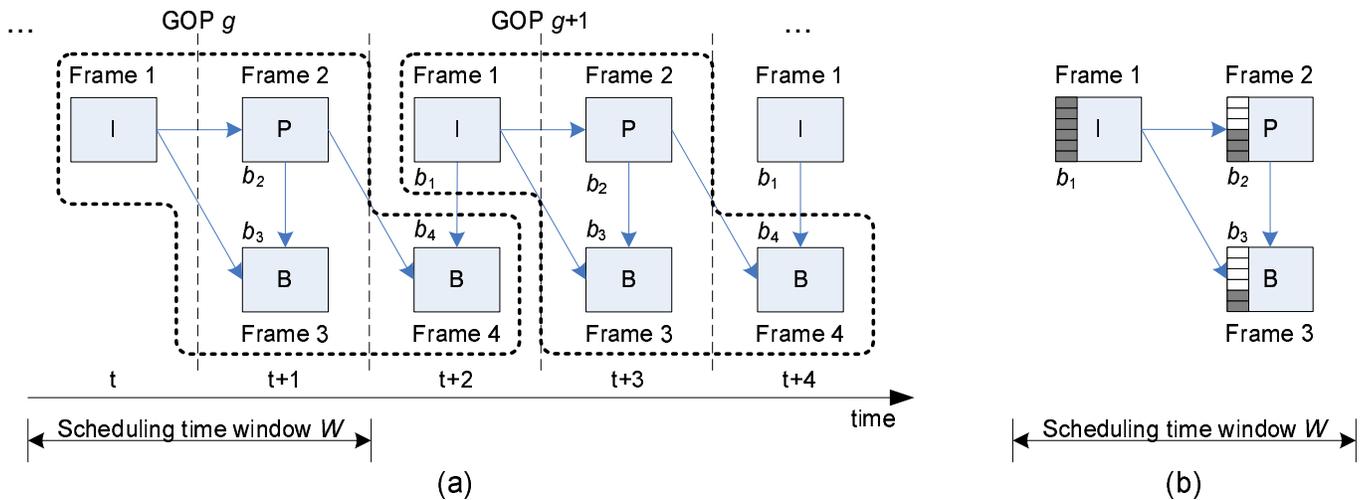}
\caption{(a) Illustrative DAG dependencies and scheduling time window using IBPB GOP structure. The schedulable frame sets defined by the scheduling time window $W$ are ${\cal F}_t=\{1,2,3\}$,
${\cal F}_{t+1}=\{2,3,4,1\}$, ${\cal F}_{t+2}=\{4,1,2,3\}$, ${\cal F}_{t+3}=\{2,3,4,1\}$, etc.
Clearly, ${\cal F}_t$ is periodic with period $T=3$ excluding the initial time $t$, and each GOP contains $N=4$ frames. (b) Traffic state detail for schedulable frame set ${\cal F}_t=\{1,2,3\}$. $b_j$ denotes the state of the $j$th frame's buffer, where $j \in {\cal F}_t=\{1,2,3\}$.}
\label{fig:traffic_state}
\end{figure}

%
%

\begin{figure}[t]
\centering
\includegraphics[width=0.85\linewidth]{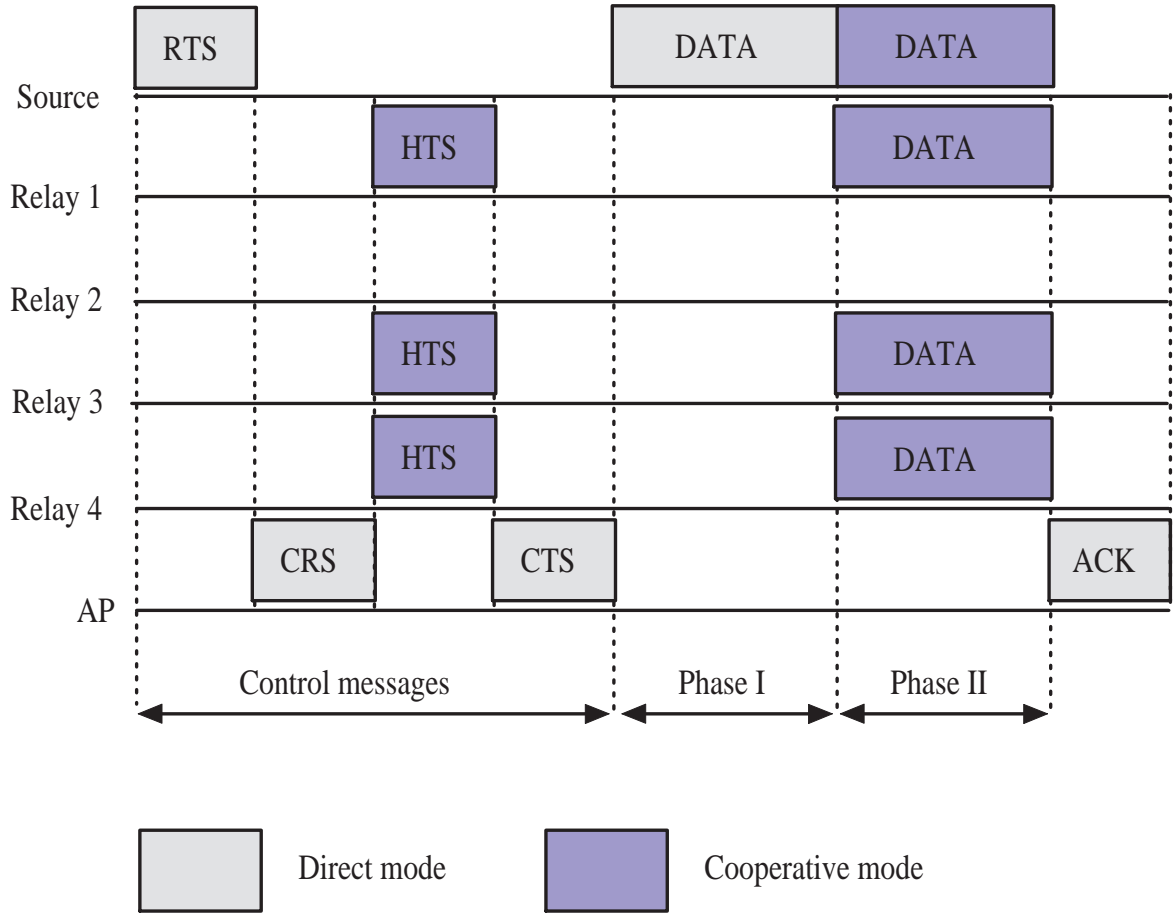}
\caption{Signaling protocol for randomized STBC cooperation.
}
\label{fig:signaling}
\end{figure}

%
%

\begin{table}
\centering
\caption{The proposed protocol for randomized STBC cooperation.}
\fbox{\hspace{0mm}
\begin{minipage}[t]{0.90\textwidth}
\vspace{3mm}

\begin{enumerate}[Step 1)]

\itemsep=0mm

\item
The $i$th source initiates the handshaking by transmitting the RTS frame, which
announces its desire to transmit data symbols and also includes training symbols 
that are used by the other nodes
to estimate the link gains. 

\item
From the RTS message, the
AP estimates the channel coefficients $h_t^{i 0}$
and, hence, determines $\beta_t^{i 0}$.
At the same time,  by passively listening to all the RTS messages
occurring in the network, the other nodes estimate their
respective channel parameters $h_t^{i \ell}$,
for $\ell \in \{1, 2, \ldots, M\}-\{i\}$, and,
thus, determine $\beta_t^{\,i \ell}$.

\item
The AP responds with the CRS message that provides feedback on $\beta_t^{i 0}$ to all
the candidate cooperative nodes and the source, as well as a second parameter 
$0 < \xi_t < 1$, which is used to recruit relays.

\item
From the CRS message, the $i$th source learns that a cooperative transmission
may take place and, if such a communication mode
will be subsequently confirmed by the AP, the data rate to be used in Phase I
is given by
\be
\rate = \frac{\beta_t^{i 0}}{\xi_t} \: .
\label{data-I}
\ee

\item
After receiving the CRS frame, 
the candidate cooperative nodes can self-select themselves according to the rule:
\begin{equation}
\label{eq.Ci}
{\mathcal C}_t^i=\left\{\ell: \frac{\beta_t^{i 0}}{{\beta}_t^{\,i \ell}}\leq \xi_t  \right\} \: ,
\end{equation}
where ${\beta}_t^{\,i \ell}$ is defined using \eqref{beta} by replacing $BEP$ with $BEP_1$.
The nodes belonging to the formed group
${\mathcal C}_t^i$ send in unison the HTS message using randomized STBC
of size $L$ as described in Section~\ref{sec:coopPHY}, which piggybacks training symbols that 
are used by the AP to estimate the
cooperative channel vector $\mathbf{R} \, {\h}_t^{i,2}$.

\item
After estimating the channel of the cooperative link, 
the AP computes the data rate $\ratee$ by resorting to \eqref{beta_2} and verifies
the fulfillment of the following condition
\be
\frac{1}{R_c \, \ratee} <  \frac{1-\xi_t}{\beta_t^{i 0}} \: .
\label{cond-2-table}
\ee
If \eqref{cond-2} holds, then, accounting also for \eqref{data-I}, it can be
inferred that cooperation is better than direct transmission, i.e., condition
\eqref{coop-useful} is satisfied: in this case, $z_t^i=1$.
Otherwise, cooperation is useless: in this case, $z_t^i=0$.
Therefore, the AP responds with a CTS
frame, which conveys the following information:
(i) the cooperation decision $z_t^i$;
(ii) if $z_t^i=1$, the data rate $\ratee$ in Phase II given by \eqref{beta_2};
(iii) the resource price $\lambda$ computed as explained in Section~\ref{sec:opt-scheduling}.

\item
If $z_t^i=1$ in the CTS frame, the source proceeds with sending in Phase I 
its data frame at rate \eqref{data-I}; otherwise, if $z_t^i=0$, it transmits in direct mode
at the data rate  $\beta_t^{i 0}$.

\item
If $z_t^i=1$ in the CTS frame, along with the source, the self-recruited relays 
cooperatively transmit in Phase II the data frame at rate $\ratee$;
otherwise, if $z_t^i=0$, they remain silent.

\item
The AP finishes the procedure by sending back to the source an acknowledgement (ACK) message.

\end{enumerate}

\vspace{3mm}
\end{minipage}
\hspace{0mm} }
\label{tab:protocol}
\end{table}

%
%

\begin{table}[t]
\centering
\caption{Simulation parameters.}
\includegraphics[width=0.45\linewidth]{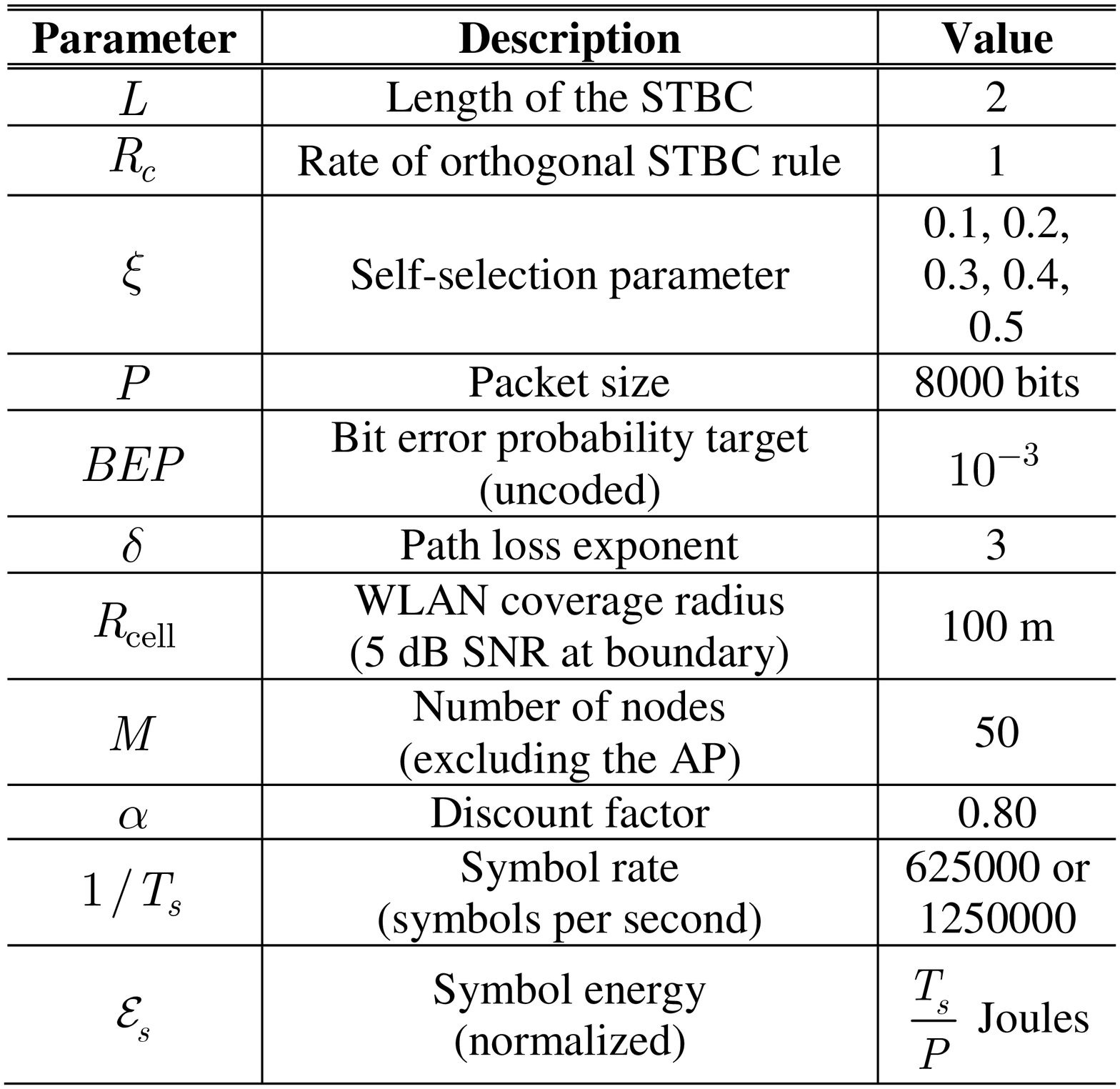}
\label{table:parameter_table}
\end{table}

%
%

\begin{figure}[t]
\centering
\includegraphics[width=0.45\linewidth]{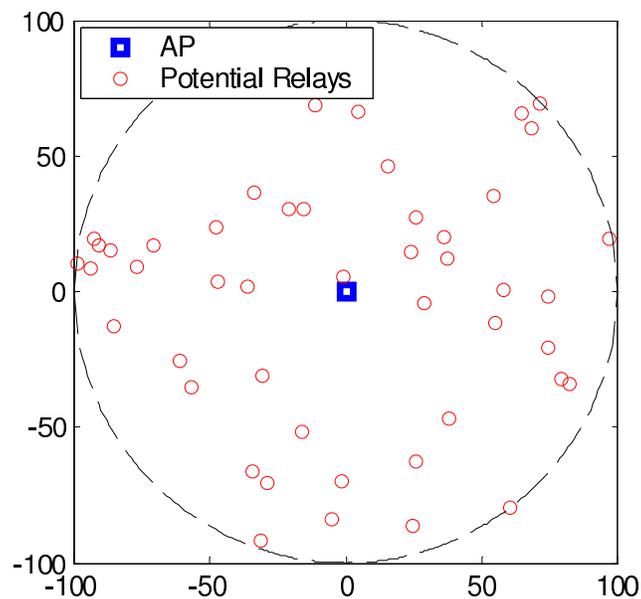}
\caption{Network topology used for numerical results. There are 50 nodes placed randomly and uniformly throughout the AP's 100 m coverage range.}
\label{fig:snapshot}
\end{figure}

%
%

\begin{figure}[t]
\centering
\includegraphics[width=0.83\linewidth]{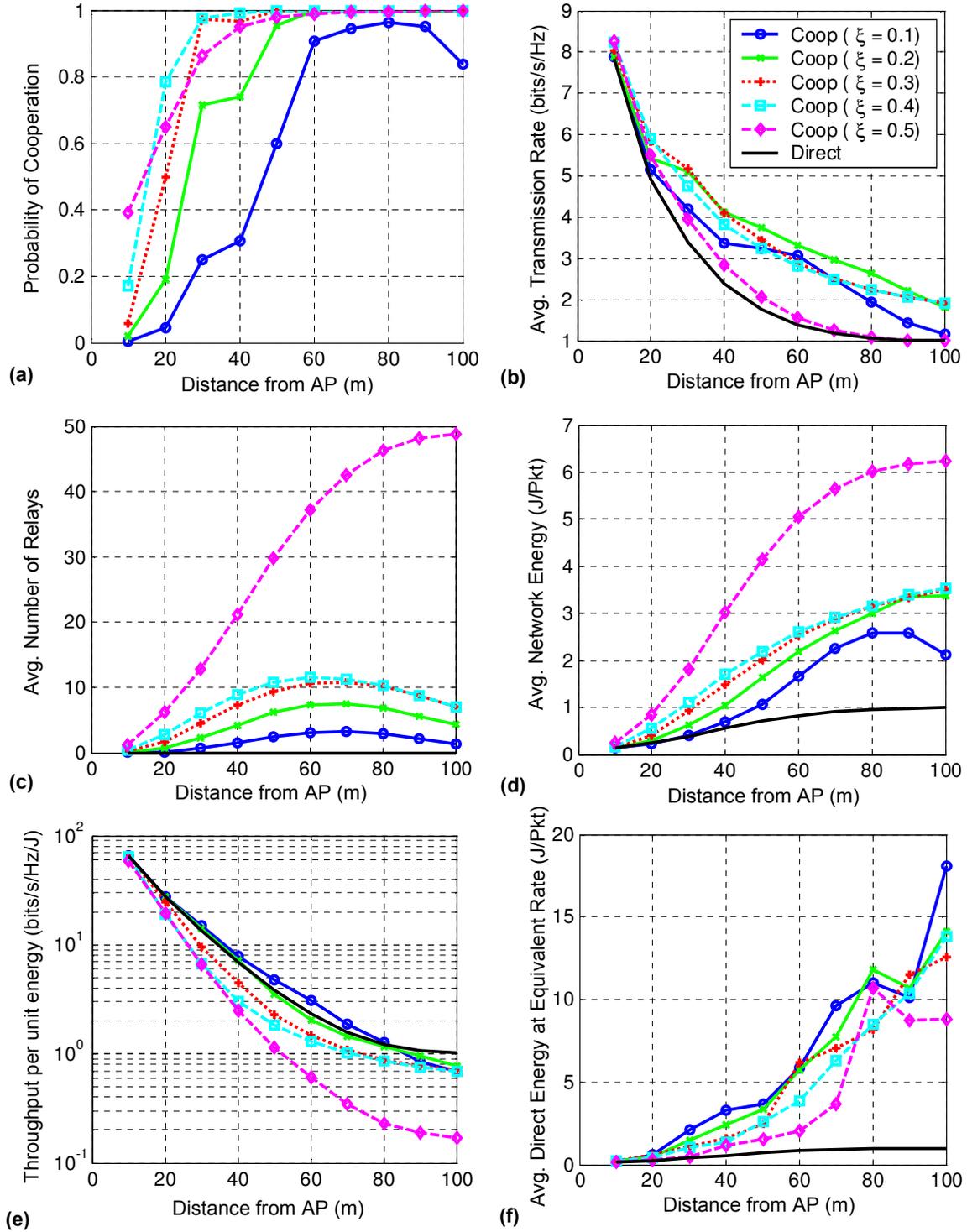}
\caption{Cooperative transmission statistics for different values of the self-selection parameter $\xi$ and for different distances from the AP. (a) Average transmission rate. (b) Probability of cooperation being optimal. (c) Average number of recruited relays. (d) Average energy consumed in the network per packet transmission. (e) Throughput per unit energy. (f) Average energy required by the source to transmit one packet at the rate $\beta_t^{i0} = \beta_t^{i,\text{coop}}$.}
\label{fig:coop_stats}
\end{figure}

%
%

\begin{figure}[t]
\centering
\includegraphics[width=0.5\linewidth]{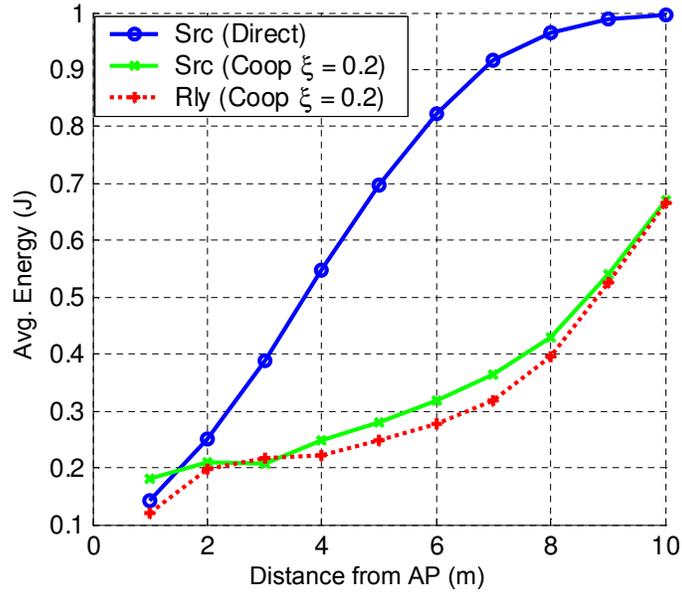}
\caption{Average energy consumed by source (Src) during direct and cooperative transmission, and average energy consumed by a relay (Rly) during cooperative transmission. A self-selection parameter $\xi=0.2$ is used for cooperative transmission.}
\label{fig:burden}
\end{figure}
%
%

\begin{figure}[t]
\centering
\includegraphics[width=1.0\linewidth]{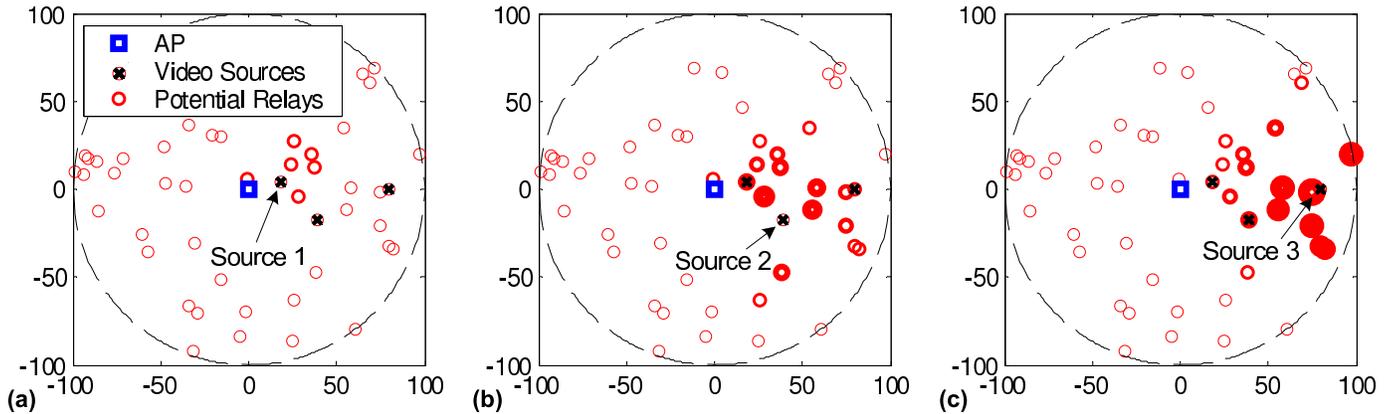}
\caption{Video source placement for homogeneous and heterogeneous streaming scenarios. Three video sources are placed 20 m, 45 m, and 80 m from the AP at angles $25^\circ$, $-30^\circ$, and $0^\circ$, respectively. (a,b,c) Relay activation frequencies for video source 1, 2, and 3, respectively, with self-selection parameter $\xi=0.2$. The size of the relay is proportional to the frequency with which it is activated as a helper for the corresponding source.}
\label{fig:relay_freqs}
\end{figure}

%
%

\begin{figure}[t]
\centering
\includegraphics[width=0.9\linewidth]{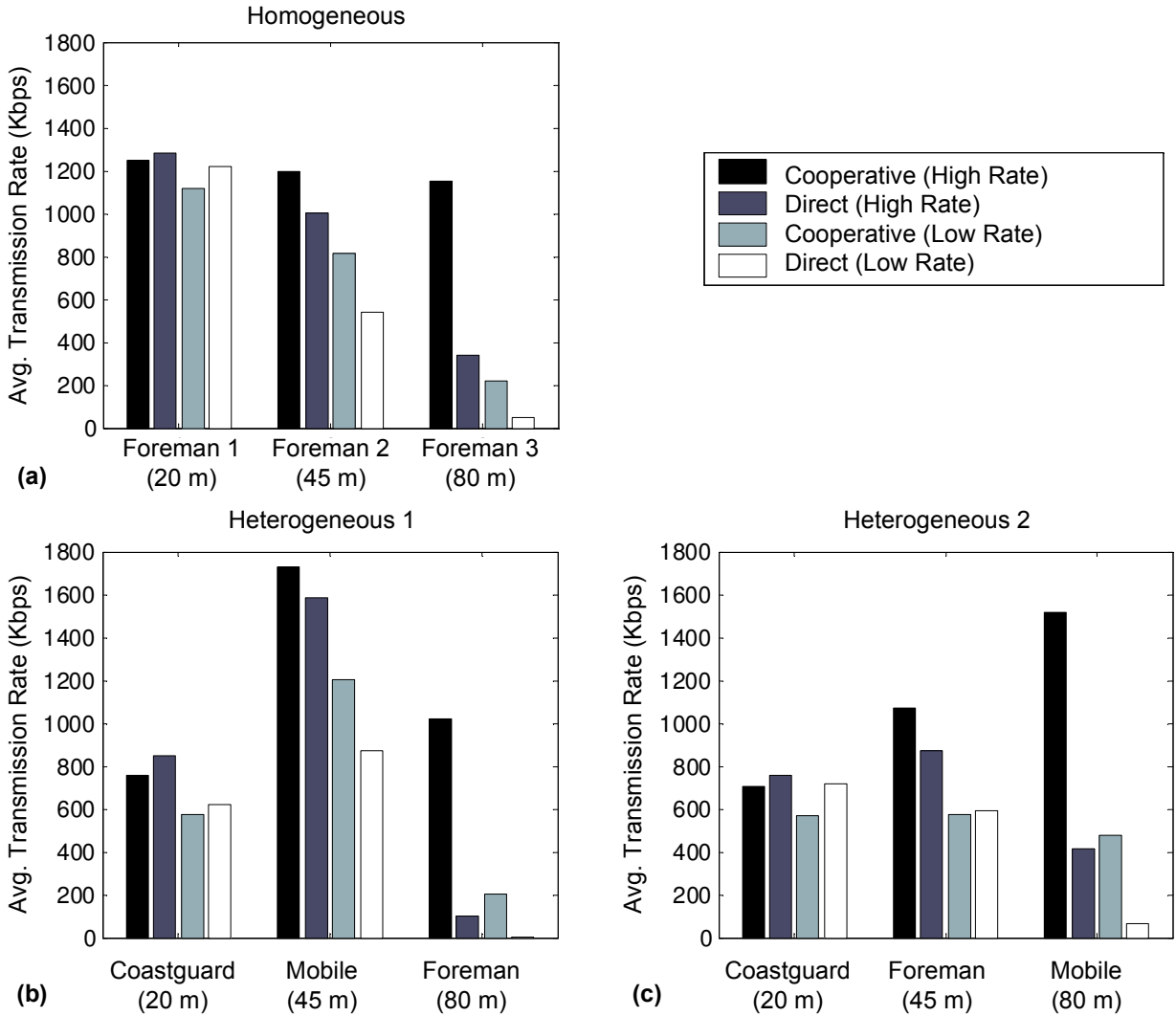}
\caption{Average transmission rates in different scenarios. (a) Homogeneous video sources. (b,c) Heterogeneous video sources.}
\label{fig:tx_rates}
\end{figure}

%
%

\begin{table}[t]
\centering
\caption{Resource prices and resource utilization in different scenarios.}
\includegraphics[width=0.6\linewidth]{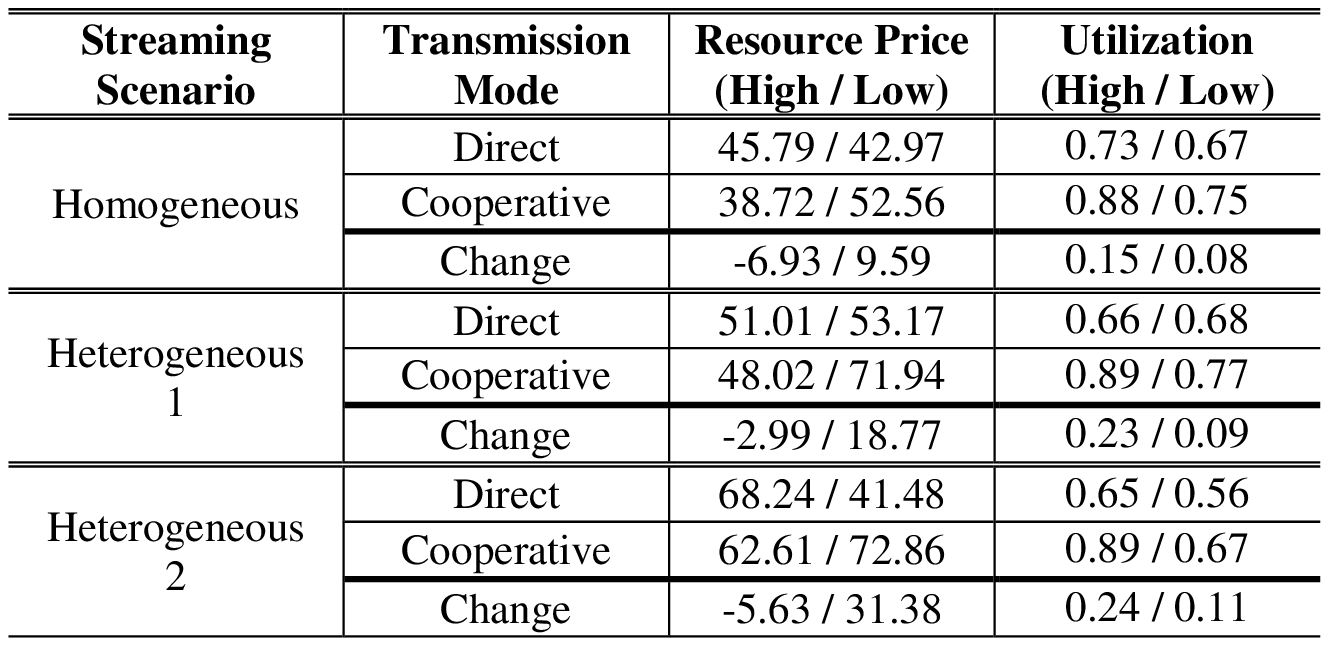}
\label{table:rsrc_price}
\end{table}

%
%

\begin{table}[t]
\centering
\caption{Expected discounted average utility in different scenarios.}
\includegraphics[width=0.9\linewidth]{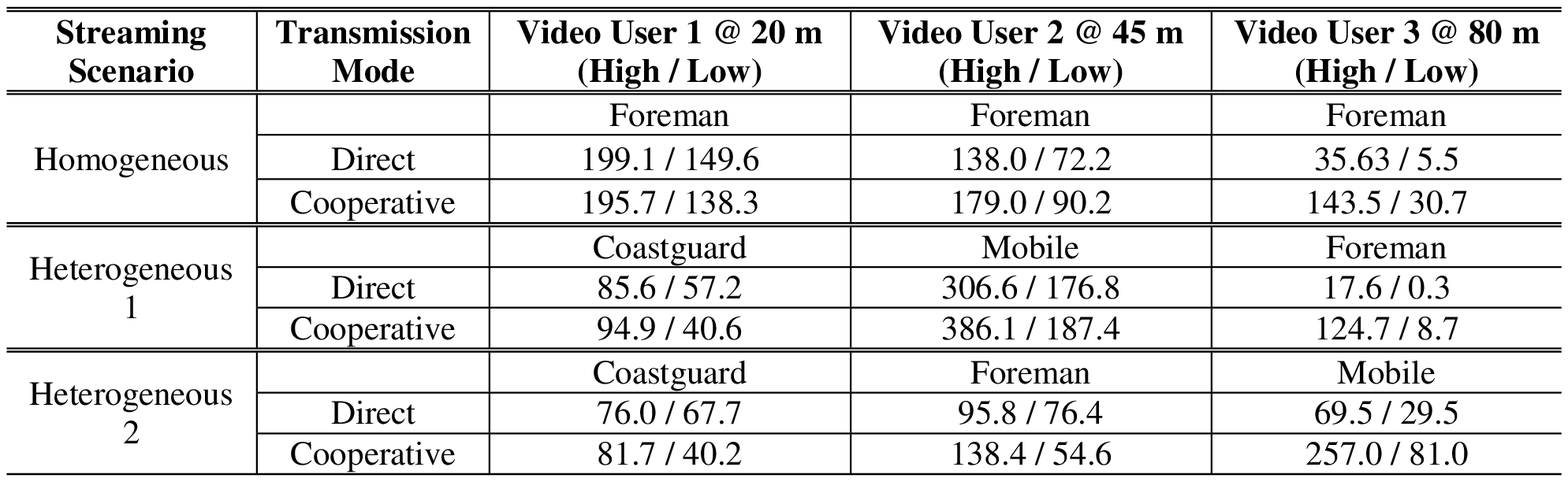}
\label{table:discounted_utility_table}
\end{table}
%
%

\begin{table}[t]
\centering
\caption{Average video quality (PSNR) in different scenarios.}
\includegraphics[width=0.9\linewidth]{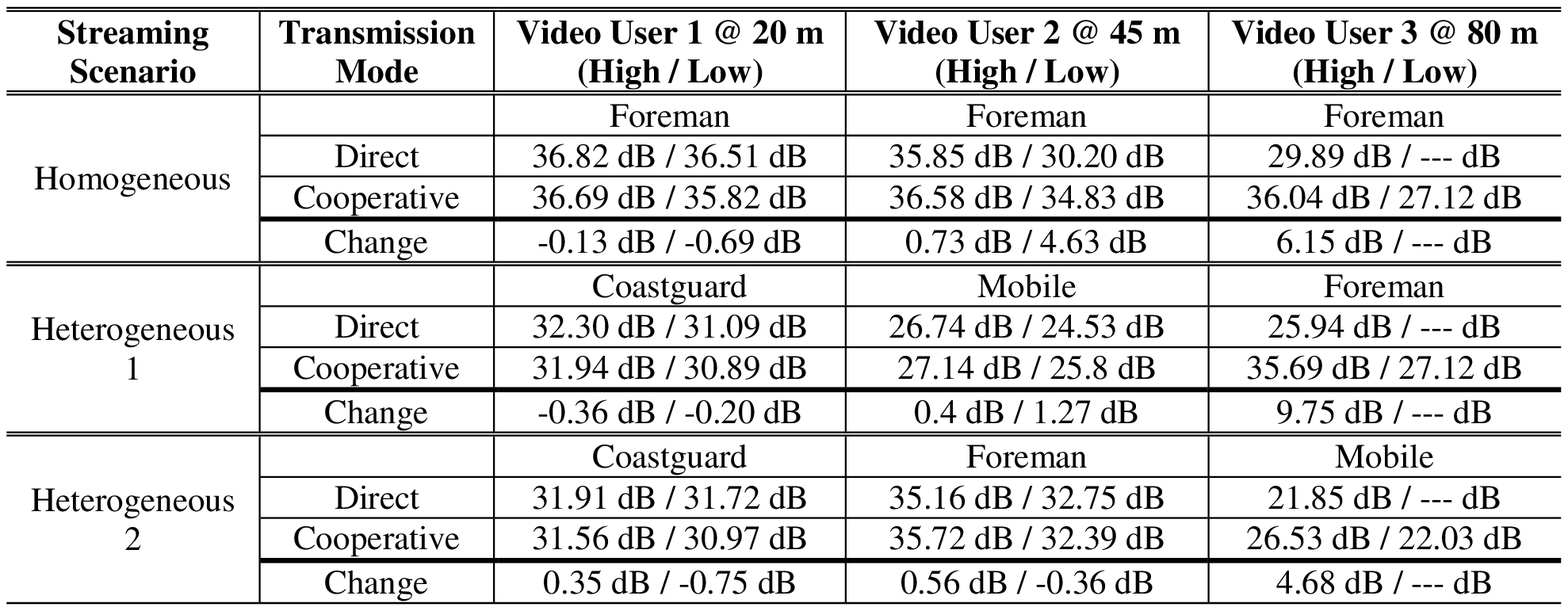}
\label{table:psnr_table}
\end{table}

\end{document}